\documentclass[reprint,jmp,superscriptaddress,showpacs,aip,onecolumn,10pt]{revtex4-1}

\advance\textwidth by -4cm
\advance\hoffset by 1.8cm
\advance\textheight by -0.5cm
\advance\voffset by 1cm

\usepackage{amsmath}
\usepackage{amssymb}
\usepackage{graphicx}
\usepackage{color}

\usepackage{amsthm}

\usepackage{hyperref}

\newtheorem{thm}{Theorem}

\newtheorem{lem}[thm]{Lemma}
\newtheorem{proposition}[thm]{Proposition}
\newtheorem{corollary}[thm]{Corollary}
\newtheorem{example}[thm]{Example}
\newtheorem{remark}[thm]{Remark}

\newcommand{\ip}[2]{\left\langle\,#1\,|\,#2\,\right\rangle}
\newcommand{\kb}[2]{|#1\,\rangle\langle\,#2|}
\newcommand{\tr}[1]{\mathrm{tr}\left[{#1}\right]}

\newcommand{\h}[1]{\mathcal{#1}}

\def\idty{{\leavevmode\rm 1\mkern -5.4mu I}} 

\def\Rl{{\mathbb R}}
\def\Ir{{\mathbb Z}}\def\Nl{{\mathbb N}}

\let\eps\varepsilon

\def\norm #1{\Vert #1\Vert}

\def\supp{{\mathop{\rm supp}\nolimits}\,}

\def\braket#1#2{\langle #1,#2\rangle}
\def\brAket#1#2{\langle #1\vert#2\rangle}
\def\brAAket#1#2#3{\langle#1\vert#2\vert#3\rangle}

\def\ket #1{\vert#1\rangle}
\def\ketbra #1#2{{\vert#1\rangle\langle#2\vert}}
\def\kettbra#1{\ketbra{#1}{#1}}

\def\tr{\mathop{\rm tr}\nolimits}
\def\abs#1{\vert#1\vert}

\def\Ell{{\mathcal L}}

\let\veps\varepsilon

\def\AA{{\mathcal A}}\def\HH{{\mathcal H}}\def\LL{{\mathcal L}}
\def\MM{{\mathcal M}}\def\NN{\mathcal N}
\def\Mav{\overline{M}} 
\def\Airy{{\rm Ai}}

\def\trcl{{\mathfrak T}}
\def\qp{{(q,p)}}
\def\Co{{\mathcal C}_0(\Rl^2)}
\def\Calt{\AA} 
\def\sab{_{\alpha\mkern1mu\beta}} 
\def\cube{[0,1]^I}
\def\cubo{\cube_{\mathrm obs}}

\begin{document}
\title{Measurement uncertainty relations}

\author{Paul Busch}
\email{paul.busch@york.ac.uk}
\affiliation{Department of Mathematics, University of York, York, United Kingdom}

\author{Pekka Lahti}
\email{pekka.lahti@utu.fi}
\affiliation{Turku Centre for Quantum Physics, Department of Physics and Astronomy, University of Turku, FI-20014 Turku, Finland}

\author{Reinhard F. Werner}
\email{reinhard.werner@itp.uni-hannover.de}
\affiliation{Institut f\"ur Theoretische Physik, Leibniz Universit\"at, Hannover, Germany}

\date{29 March 2014 (Revised Version)}
\begin{abstract}\narrower{\small
Measurement uncertainty relations are quantitative bounds on the errors in an approximate joint measurement of two observables. They can be seen as a generalization of the error/disturbance tradeoff first discussed heuristically by Heisenberg. Here we prove such relations for the case of two canonically conjugate observables like position and momentum, and establish a close connection with the more familiar preparation uncertainty relations constraining the sharpness of the distributions of the two observables in the same state. Both sets of relations are generalized to means of order $\alpha$ rather than the usual quadratic means, and we show that the optimal constants are the same for preparation and for measurement uncertainty. The constants are determined numerically and compared with some bounds in the literature. In both cases the near-saturation of the inequalities entails that the state (resp. observable) is uniformly close to a minimizing one.\\

\noindent{\small Published in:  {\em J. Math. Phys.} {\bf 55} (2014) 042111. 
\href{http://dx.doi.org/10.1063/1.4871444}{DOI:10.1063/1.4871444}}

}
\end{abstract}

\maketitle

\section{Introduction}
Following Heisenberg's ground-breaking paper \cite{Heisenberg1927} from 1927 uncertainty relations have become an indispensable tool of quantum mechanics. Often they are used in a qualitative heuristic way rather than in the form of proven inequalities. The paradigm of the latter kind is the Kennard-Robertson-Weyl inequality\cite{Kennard,Robertson,Weyl}, which is proved in virtually every quantum mechanics course. This relation shows that by turning the uncertainty relations into a theorem, one also reaches a higher level of conceptual precision. Heisenberg talks rather vaguely of quantities being ``known'' to some precision. With Kennard, Robertson and Weyl we are given a precise physical setting: a position and a momentum measurement applied to distinct instances of the same preparation (state), with the uncertainties interpreted as the root of the variances of the probability distributions obtained. However, this mathematical elucidation does not cover all quantitative aspects of uncertainty. There have been several papers in recent years formalizing and proving further instances of quantitative uncertainty. This is perhaps part of a trend which has become necessary as more and more experiments approach the uncertainty-limited regime. Moreover, uncertainty relations play an important role in some proofs of quantum cryptographic security.  For a review of uncertainty up to 2006 we recommend \cite{BuHeLa07}.

The Kennard-Robertson-Weyl inequality can be modified in several ways: On the one hand, we can stick to the same physical scenario, and apply different definitions of ``spread'' of a probability distribution. This is one of the routes taken in this paper: We replace the quadratic mean by one based on powers $\alpha$ and $\beta$ for the position and the momentum distributions, respectively. One could go further here and introduce measures of ``peakedness'' for probability distributions such as entropies \cite{Hirschman}. Our main goal, however, is a modification of the basic scenario, much closer to the original discussion of Heisenberg. In his discussion of the $\gamma$-ray microscope the uncertainty relations concern the resolution of the microscope and the momentum kick imparted by the Compton scattering. Due to this momentum kick a momentum measurement after the observation gives different results from a direct momentum measurement. We generalize this scenario by taking the microscope with a subsequent momentum measurement as a ``phase space measurement''. i.e., as an observable which produces a position value and a momentum value in every single shot. Measurement uncertainty relations then constrain the accuracy of the marginal measurements: the resolution of the microscope is a benchmark parameter for the position output of the device, and the deviation from an  ideal position measurement. Similarly, the disturbance is a quantity characterizing the accuracy of the momentum output, once again as compared to an ideal momentum measurement. Then, by definition, a measurement uncertainty relation is an inequality implying that these two error quantities cannot both be small. It has been disputed recently in a series of papers \cite{Ozawa03,Ozawa04,Ozawa05,Erh12,Roz12} that such a relation holds. We believe that this claim is based on badly chosen definitions of uncertainty, and we will give a detailed critique of these claims elsewhere \cite{BLW2013a}.

We offer two ways of quantifying the error quantities in the measurement uncertainty relation: on the one hand by a calibration process, based on worst case deviations of the output distribution, when the device is presented with states of known and sharp position (resp.\ momentum). On the other hand we introduce a distance of observables based on transportation metrics. In both definitions a power $\alpha$ can naturally be used. Thus for preparation uncertainty as well as both definitions of errors for measurement uncertainty we will prove relations of the form
\begin{equation}\label{URab}
  \Delta_\alpha(P)\Delta_\beta(Q)\geq c\sab\, \hbar,
\end{equation}
of course with quite different interpretations of the $\Delta$-quantities. The constants $c\sab$, however, will be the same and best possible in all three cases. The cases of equality can be characterised precisely (they depend on $\alpha$ and $\beta$). Moreover, there is a second constant $c'\sab$ such that an uncertainty product $u$ with $c\sab\hbar\leq u<c'\sab\hbar$ implies that the state (in the case of preparation uncertainty) or the observable (in the case of measurement uncertainty) is uniformly close to one with strictly minimal uncertainty, with an error bound going to zero as \eqref{URab} becomes sharp.

The basic idea of measurement uncertainty and the idea of calibration errors was presented, together with a sketch of the proof, in a recent letter \cite{BLW2013c}. Here we give a full version of the proof. At the same time we lift the restriction $\alpha=\beta=2$, thereby covering also the previously studied case $\alpha=\beta=1$ based only on the metric uncertainty definition \cite{Werner04}.

Our paper is organized as follows: In Section \ref{sec:obs} we review the concept of an observable as an object that assigns
to all states the probability measures for the outcomes of a given measurement. We use this language to describe the idea of an approximate joint measurement of noncommuting observables, such as position and momentum, for which the traditional perspective of an observable as a selfadjoint operator is inadequate. We then recall the definition and some relevant properties of a covariant phase space observable, which constitutes a fundamentally important special case of an approximate joint measurement.  Section \ref{sec:errors} presents measures of measurement errors, understood as the difference between the observable actually
measured in a measurement scheme and the target observable. This difference is quantified in terms of the so-called
Wasserstein distance of order $\alpha$ for probability measures.  Maximising this distance over all pairs of probability distributions
of the estimator and target observables for the same state  yields a metric distance between observables.
If the maximization is performed over calibrating states only, we obtain a measure of calibration error for observables.

In Section \ref{sec:QPUR} we formulate and prove our main result -- the joint measurement and error-disturbance relations
for position and momentum. These inequalities are obtained as consequences of preparation uncertainty relations for these quantities,
where the usual product of standard deviations is replaced with more general choices of $\alpha$-deviations for the position
and momentum, respectively. The tight lower bound, which will be determined explicitly and investigated numerically in
Section \ref{sec:constants}, is given by Planck's constant $\hbar$ multiplied by a constant that depends on the choice of deviation measure.
In Section \ref{sec:extensions} some possible extensions and generalizations are briefly discussed.
Section \ref{sec:conclusion} concludes with a summary and outlook.

\section{Observables}\label{sec:obs}
The setting of this paper is standard quantum mechanics of a single canonical degree of freedom. In this section we fix some notation and terminology.

An {\it observable} is the mathematical object describing a measuring device. This description must provide, for every input state given by a density operator $\rho$, the probability distribution of the measurement outcomes. The set $\Omega$ of possible outcomes is part of the basic description of the device, and in order to talk about probabilities it must come equipped with a $\sigma$-algebra of ``measurable'' subsets of $\Omega$, which we suppress in the notation. The only cases needed in this paper are $\Omega=\Rl$ (position or momentum) and $\Rl^2$ (phase space), or subspaces thereof, each with the Borel $\sigma$-algebra. For the observable $F$ we denote the outcome probability measure on $\Omega$ in the state $\rho$ by $F_\rho$. Since, for every measurable $X\subset\Omega$, $\rho\mapsto F_\rho(X)$ must be a bounded linear functional, there is a positive operator $F(X)$ such that $F_\rho(X)=\tr\rho F(X)$. The measure property of $F_\rho$ then implies an operator version thereof, i.e., for a sequence of disjoint $X_i$ we have $F(\bigcup X_i)=\sum_iF(X_i)$, where the sum converges in the weak, strong and ultraweak operator topologies. We take all observables to be normalized, i.e., $F(\Omega)=\idty$. An observable is thus (given by) a normalized positive operator valued measure on the outcome space of a measurement.

When all the operators $F(X)$ are projection operators we say that $F$ is a {\it sharp} observable. The prime example is the unique spectral measure on $\Rl$ associated with a selfadjoint operator $A$, which we denote by $E^A$, and which is uniquely determined by $A$ through the resolution formula $A=\int x\,dE^A(x)$. Most textbooks use the term ``observable'' synonymously with ``selfadjoint operator'' and go on to explain how to  determine the outcome probabilities by using the spectral measure $E^A$ or, equivalently, how expectation values of functions of the outcomes are to be computed as the expectations of functions of the operator in the functional calculus. For the purposes of this paper (and many other purposes) this view is too narrow, since the phase space observables describing an approximate joint measurement of position and momentum cannot be sharp. However, the ``ideal'' position and momentum measurements will be described by the usual sharp observables $E^Q$ and $E^P$.

The principal object we study are joint measurements of position and momentum. These are observables with two real valued outcomes (i.e., $\Omega=\Rl^2$), where the first outcome is called the position outcome and the second the momentum outcome. These labels have no significance, except that we will compare the first outcomes with those of a standard position measurement and the second to those of a standard momentum observable. More precisely, we denote by $M^Q$ the first ``position'' marginal of the observable (i.e., $M^Q(X)=M(X\times\Rl)$) and ask to what extent we can have $M^Q\approx E^Q$, and at the same time $M^P\approx E^P$ for the second marginal $M^P$. The precise interpretation of the approximation will be discussed in Sect.~\ref{sec:distobs}.  This rather abstract approach covers many concrete implementations, including, of course, the scenario of Heisenberg's microscope, where an approximate position measurement is followed by a standard momentum measurement. The quality of the approximation $M^Q\approx E^Q$ is then quantified by the ``error'' of the measurement or, put positively, by the resolution of the microscope. The approximation $M^P\approx E^P$ compares the momentum after the measurement with the direct momentum measurement, i.e., the measurement without the microscope. The approximation error here quantifies the ``momentum disturbance''. However, we emphasize that the joint measurement need not  be constructed in this simple way. For example, we could make any suitable measurement on the particle after the position measurement, designed to make the approximation  $M^P\approx E^P$ as good as possible. For this we could use the position outputs and everything we know about the construction of the microscope, correcting as much as possible the systematic errors introduced by this device. So ``momentum disturbance'' is not just a question whether the direct momentum measurement after the microscope still gives a good result, but whether there is any way at all of retrieving the momentum from the post-measurement state. Needless to say, the joint measurement perspective also restores the symmetry between position and momentum, i.e., the results apply equally to an approximate momentum measurement disturbing the position information. In any case, the results will be quantitative bounds expressing that $M^Q\approx E^Q$ and $M^P\approx E^P$ cannot both hold with high accuracy.

\subsection{Covariant Phase space observables}\label{sec:covar}
Covariant phase space observables are of fundamental importance for this study. Though extensively studied in the literature  we also recall briefly  their definitions and characteristic properties. For details, see, for instance \cite{Davies,Holevo,QHA,Husimi,Cassinelli2003,KLY2006}.

In his famous paper Heisenberg announced that he would show ``a direct mathematical connection'' between the uncertainty relation  and the commutation relations of position and momentum (and henceforth forgets this announcement). Of course, this is what we will do. Due to von Neumann's uniqueness theorem for the commutation relations we may as well start from the usual form of the operators: $Q$ is the operator of multiplication by $x$ on ${\mathcal L}^2(\Rl,dx)$, and $P$ is the differentiation operator $P\psi=-i\psi'$. Here and in the sequel we will set $\hbar=1$. The joint translation by $q$ in position and by $p$ in momentum are implemented  by the Weyl operators (also known as Glauber translations) $W\qp=\exp(ipQ-iqP)$ acting explicitly as
\begin{equation}\label{Weylop}
 (W(q,p)\psi)(x)=e^{\textstyle -{i}{qp}/2+ipx}\, \psi(x-q).
\end{equation}
Of course, these commute only up to phases, which are again equivalent expression of the commutation relations.

A {\it covariant observable} is defined as an observable $M$ with phase space outcomes ($\Omega=\Rl^2$) such that, for any measurable $Z\subset\Rl^2$, 
\begin{equation}\label{cov}
  M\bigl(Z-\qp\bigr)=W\qp^*M(Z)W\qp.
\end{equation}
This implies \cite{Holevo,QHA,Cassinelli2003,KLY2006} that the measure $M$ has an operator valued density with respect to Lebesgue measure, that the densities at different points are connected by the appropriate phase space translations, and that the density at the origin is actually itself a density operator $\sigma$, i.e., a positive operator with unit trace. Explicitly we get the formula
\begin{equation}\label{povm}
  M(Z)=\int_{\qp\in Z}\mskip-5mu  W(q,p)\Pi\sigma\Pi W(q,p)^* \,\frac{dq\,dp}{2\pi\hbar}\ .
\end{equation}
Here we deviated from the announcement to set $\hbar=1$ to emphasize that the measure for which the density of a normalized observable has trace $1$ is the usual volume normalization in units of Planck's ``unreduced'' constant $h=2\pi\hbar=2\pi$. The operator $\Pi$ is the parity operator $(\Pi\psi)(x)=\psi(-x)$, and merely changes the parametrization of observables in terms of $\sigma$. The reason for this will become clear presently.

We need to compute the expectations $M_\rho^Q$ of the $Q$-marginal of a covariant observable. For the sake of this computation we may set $\rho=\kettbra\psi$ and $\sigma=\kettbra\phi$, and later extend by linearity. Then the probability density of $M_\rho^Q$ at the point $q$ is obtained from \eqref{povm} by taking the expectation with $\rho$ and  leaving out the integral over $q$ while retaining the one over $p$. This gives
\begin{equation}\label{MQrhoDensity}\begin{split}
   \frac1{2\pi}&\int \braket{\psi}{W\qp\Pi\phi}\braket{W\qp\Pi\phi}{\psi}\ dp\\
    &=\frac1{2\pi}\int \overline{\psi(x)}\psi(y)\phi(q-x)\overline{\phi(q-y)}\ e^{i(px-py)}\ dp\,dx\,dy
   = \int \abs{\psi(x)}^2\ \abs{\phi(q-x)}^2\ dx \ .
\end{split}
\end{equation}
Here we used $\int e^{ipx}dp=2\pi\delta(x)$. The result is the convolution of the position distributions of $\rho$ and $\sigma$. Together with the analogous relation for momentum we can write this as
\begin{equation}\label{marginalconcolve}
  M_\rho^Q=E^Q_\rho \ast E_\sigma^Q  \quad\mbox{and } M_\rho^P=E^P_\rho \ast E_\sigma^P\ ,
\end{equation}
where the star denotes the convolution of measures or their density functions. Thus we arrive at the key feature of covariant measurements for our study: {\it The marginal distributions of a covariant measurement are the same as those of the corresponding ideal measurement with some added noise, which is statistically independent of the input state.} The noise distributions are just $E_\sigma^Q$ and  $E_\sigma^P$, so they are constrained exactly by the preparation uncertainty of $\sigma$.

\begin{remark} \rm
There is a converse to Eq.\eqref{marginalconcolve}. Rather than asking how we can approximately measure the standard position and momentum observables together, we can ask under which conditions approximate position and momentum observables can exactly be measured together. Here by an approximate position measurement we mean an observable $F$, which is covariant for position shifts, and commutes with momentum, so that $F(X-q)=W\qp F(X)W^*\qp$. These are necessarily of the form $F_\rho=\mu\ast E^Q_\rho$ for some measure $\mu$ (see, e.g., \cite{CHT2004}). Suppose that we have such an approximate position measurement and, similarly, an approximate momentum measurement given by a noise measure $\nu$. Then, using the averaging technique developed in\ \cite{Werner04} (reviewed below), it can also be shown that these two are jointly measurable, i.e., they are the marginals of some phase space observable $M$, if and only if there is a covariant observable $M$ with these marginal, i.e., if and only if $\mu=E^Q_\sigma$ and $\nu=E^P_\sigma$ for some density operator $\sigma$.\cite{CHT2005}
\end{remark}

\begin{remark} \rm
The parity operator appears either in \eqref{povm} or in \eqref{marginalconcolve}. The convention we chose is in agreement with the extension of the convolution operation from classical measures on phase space to density operators and measures. Indeed, a convolution can be read as the average over the translates of one factor weighted with the other. Therefore the convolution of a density operator $\rho$ and a probability measure $\mu$ on phase space is naturally defined as the density operator
\begin{equation}\label{frhoconvolve}
  \mu\ast\rho=\rho\ast\mu= \int W\qp\rho W^*\qp\ d\mu\qp\ .
\end{equation}
This sort of definition would work for any group representation. A special property of the Weyl operators (``square integrability'') allows us to define \cite{QHA} a convolution also of two density operators giving the probability density
\begin{equation}\label{rhosigmaconvolve}
  \bigl(\rho\ast\sigma\bigr)\qp= \tr\rho\ W\qp\Pi\sigma\Pi W^*\qp\ .
\end{equation}
It turns out that the integrable functions on phase space together with the trace class then form a commutative and associative Banach algebra or, more precisely, a $\Ir_2$-graded Banach algebra, where functions have grade $0$, operators have grade $1$, and the grade of a product is the sum of the grades mod~$2$. Since this algebra is commutative, it can be represented as a function algebra, which is done by the Fourier transform for functions and by the Fourier-Weyl transform $({\mathcal F}\rho)\qp=\tr\rho W\qp$ for operators. The Wigner function of $\rho$ is then the function (not usually integrable) that has the same Fourier transform as $\rho$. This explains why the convolution of two Wigner functions is positive: this is just the convolution of the density operators in the sense of \eqref{rhosigmaconvolve}. A similar argument will be used in the proof of Prop.~\ref{thm:Ncompact}.

With this background the appearance of convolutions in \eqref{marginalconcolve} is easily understood: It is just the equation $M_\rho=\rho\ast\sigma$ for the phase space density, integrated over $p$ and $q$ respectively. Similarly, it becomes clear that any kind of variance of the observed joint distribution $M_\rho$ will be the sum of two terms, one coming from the preparation $\rho$ and one coming from the measurement defined by $\sigma$, but that these two play interchangeable roles.
\end{remark}

\section{Quantifying measurement errors}\label{sec:errors}

In a fundamentally statistical theory like quantum mechanics the results of individual measurements tell us almost nothing: It is always the probability distribution of outcomes for a fixed experimental arrangement which can properly be called the result of an experiment. The fact that even  for a pure state  ($\rho=\kb\psi\psi$) the probabilities $\ip{\psi}{F(X)\psi}$ usually take  values other than 0 or 1 is not a bug but a feature of quantum mechanics. Therefore the variance of position in a particular state has nothing to do with an ``error'' of measurement. There is no ``true value'' around which the outcomes are scattering, and which the measurement is designed to uncover. The variance merely provides some partial information about the probability distribution and a careful experimenter will record as much information about this distribution as can be reliably inferred from his finite sample of individual measurements.

Nevertheless, experimental errors occur in this process. However, they cannot be detected from just one distribution. Instead they are related to a difference between the observable the experimenter tries to measure and the one that she actually measures. When the state is fixed, this amounts to a difference between two probability distributions. In this section we will review some ways of quantifying the distance between probability distributions. For the sake of discussion let us take two probability measures $\mu$ and $\nu$ on some set $\Omega$.

\begin{remark}\rm
For a probability measure $\mu$ on the real line one may determine its moments $\mu[x^n]=\int x^n\,d\mu(x)$, $n=0,1,2,\ldots$. Even if all the moment integrals exist and are finite they do not necessarily determine the probability measure.
Thus, on the statistical level of moments, two probability measures $\mu$ and $\nu$ may be indistinguishable even when $\mu\ne\nu$.
This is not a mathematical artefact but a rather common quantum mechanical situation. Indeed, consider, for instance, the double-slit states defined by the functions
$\psi_\delta=\frac 1{\sqrt{2}}(\psi_1+e^{i\delta}\psi_2)$, $\delta\in\Rl$, where $\psi_1,\psi_2$ are  smooth functions with disjoint compact supports (in the position representation). A direct computation shows that the moments of  the momentum distribution $E^P_{\psi_\delta}$ are independent of $\delta$ although the distribution $p\mapsto |\hat\psi_\delta(p)|^2
=\frac 12[|\hat\psi_1(p)|^2+|\hat\psi_2(p)|^2+2{\rm Re}(\overline{\hat\psi_1(p)}\hat\psi_2(p)e^{i\delta})]$ is $\delta$-dependent. Therefore, the moments do not distinguish between the different  distributions $E^P_{\psi_\delta}$ and $E^P_{\psi_\gamma}$, $\delta\neq\gamma ({\rm mod}\ 2\pi)$. This is to remind us that a discrimination between two probability measures cannot, in general, be obtained from moments alone; in particular,  expectations $\mu[x]$ and variances $\mu[x^2]-\mu[x]^2$ are not enough. If $\mu$ is compactly supported or exponentially bounded, then the moments $(\mu[x^n])_{n\geq 0}$ uniquely determine $\mu$ (see, for instance,\ \cite{Simon1998}).

\end{remark}

\subsection{Variation norm}
The most straightforward distance measure is the {\it variation norm}, which is equal to the ${\mathcal L}^1$ distance of the probability densities when the two measures are represented by densities with respect to a reference measure. Operationally, it is twice the largest difference in probabilities:
\begin{equation}\label{varnorm}
  \norm{\mu-\nu}_1=2\sup_{X\subset\Omega}\bigl|\mu(X)-\nu(X)\bigr|
              =\sup\Bigl\{\bigl|{\textstyle\int f(x)d\mu(x)-\int f(x)d\nu(x)}\bigr|\,\Bigm| -1\leq f(x)\leq1\Bigr\}
\end{equation}
For observables $E$ and $F$ we consider the corresponding norm
\begin{equation}\label{varnormobs}
  \norm{E-F}=\sup_{\rho}\norm{E_\rho-F_\rho}_1
            =\sup_{\rho,f}\tr\rho\Bigl(\int f(x)dE(x)-\int f(x)dF(x)\Bigr),
\end{equation}
where the sup runs over all density operators $\rho$ and all measurable functions $f$ with $-1\leq f(x)\leq1$. Thus the statement ``$\norm{E-F}\leq\varepsilon$'' is equivalent to the rather strong claim that no matter what input state and outcome event we look at, the probability predictions from $E$ and $F$ will never differ by more than $\varepsilon/2$.

However, this measure of distance is not satisfactory for quantifying measurement errors of continuous variables. Indeed there is no reference to the distance of underlying points in $\Omega$. Thus two point measures of nearby points will have distance $2$, no matter how close the points are. Another way of putting this is to say that variation distance is dimensionless like a probability. What we often want, however, is a distance of probability distributions, say, on position space, which is measured in meters. The distance of two point measures would then be the distance of the points, and shifting a probability distribution by $\delta x$ would introduce an ``error'' of no more than $\abs{\delta x}$. It is clear that this requires a metric on the underlying space $\Omega$, so from now on we assume a metric $D:\Omega\times\Omega\to\Rl_+$ to be given. Of course, in the case of $\Rl$ or $\Rl^n$ we just take the Euclidean distance $D(x,y)=\abs{x-y}$.

\subsection{Metric distance from a point measure}

Let us begin with a simple case, which is actually already sufficient for preparation uncertainty and for the calibration definition of measurement uncertainty: We assume that one of the measures is a point measure, say $\nu=\delta_y$. Then, for $1\leq\alpha<\infty$, we define the  {\em deviation of order $\alpha$}, or $\alpha$-{\em deviation}, of $\mu$ from $y$ as
\begin{equation}\label{DelalphaPoint}
  D_\alpha(\mu,\delta_y)=\left(\int D(x,y)^\alpha\,d\mu(x)\right)^{\frac1\alpha}.
\end{equation}
The letter $D$ is intentionally chosen to be the same: This definition will be an instance of the general extension of the underlying metric $D$ from points to probability measures. Note that, in particular, we have $D_\alpha(\delta_x,\delta_y)=D(x,y)$ for all $\alpha,x,y$. We will also consider the limiting case $\alpha=\infty$, for which we have to set
\begin{equation}\label{DelalphaInfty}
  D_\infty(\mu,\delta_y)=\mu-{\rm ess} \sup\{D(x,y)| x\in\Omega\}=\inf\Bigl\{t\geq 0\Bigm| \mu\{(x,y)|D(x,y)>t\}=0\Bigr\}.
\end{equation}
Of course, any one of the expressions \eqref{DelalphaPoint}, \eqref{DelalphaInfty} may turn out to be infinite.

Connected to the $\alpha$-deviations are the {\it $\alpha$-spreads} or minimal deviations of order $\alpha$, namely
\begin{equation}\label{Delalpha}
  \Delta_\alpha(\mu)=\inf_{y\in\Rl}D_\alpha(\mu,\delta_y)=\inf_y\left(\int\abs{x-y}^\alpha\ d\mu(x),\right)^{\frac1\alpha}
\end{equation}
where the second expression (valid for $1\leq\alpha<\infty$) just inserts the definition of $D_\alpha$ for the only metric space ($\Omega=\Rl$, absolute value distance), which we actually need in this paper. When \eqref{DelalphaPoint} is interpreted as distance, \eqref{Delalpha} represents the smallest distance of $\mu$ to the set of point measures. In the case of $\Omega=\Rl$ and $\alpha=2$ we recover the ordinary standard deviation, since the infimum is attained for $y$ equal to the mean. The point $y$ to which a given measure is ``closest'' depends on $\alpha$. For the absolute deviation ($\alpha=1$) this is the median, for $\alpha=2$ it is the mean value, and for $\alpha=\infty$ it is the midpoint of the smallest interval containing the support $\supp(\mu)$ of $\mu$, the smallest closed set of full measure.

The interpretation of \eqref{DelalphaPoint}/\eqref{DelalphaInfty} as ``distance to a point measure'' hinges on the possibility to extend this definition to a metric proper on the set of probability measures. This is done in the following section.

\subsection{Metric distance for probability distributions}

The standard distance function with the properties described above is known as the Monge-Kantorovich-Wasserstein-Rubinstein-Ornstein-Gini-Dall'Aglio-Mallows-Tanaka\linebreak ``earth mover's'' or ``transportation'' distance, or some combination of these names (see \cite{Villani} for background and theory, for $\alpha=\infty$  we refer to \cite{Champion2008,Jylha2013}). For purpose of assessing the accuracy of quantum measurements it was apparently first used by Wiseman \cite{Wiseman1998}.
The natural setting for this definition is an outcome space $\Omega$, which is a complete separable metric space with metric $D:\Omega\times\Omega\to\Rl_+$.
For any two probability measures $\mu,\nu$ on $\Omega$ we define a {\it coupling} to be a probability measure $\gamma$ on $\Omega\times \Omega$  with  $\mu$ and $\nu$ as the marginals.
We denote by $\Gamma(\mu,\nu)$ the set of couplings between $\mu$ and $\nu$. Then, for any $\alpha$, $1\leq \alpha<\infty$ we define the $\alpha$-distance (also Wasserstein $\alpha$-distance \cite{Villani}) of $\mu$ and $\nu$ as
\begin{equation}\label{Wasserstein}
  D_\alpha(\mu,\nu)=  \inf_{\gamma\in\Gamma(\mu,\nu)}  D^\gamma_\alpha(\mu,\nu)=
  \inf_{\gamma\in\Gamma(\mu,\nu)}\left(\int D(x,y)^\alpha\,d\gamma(x,y) \right)^{\frac1\alpha}
\end{equation}
For $\alpha=\infty$, one again defines $D^\gamma_\infty(\mu,\nu)=\gamma-{\rm ess}\ \sup\{D(x,y)\,|\, (x,y)\in\Omega\times\Omega\}$ and thus
\begin{equation}\label{Wassersteininfty}
   D_\infty(\mu,\nu) = \inf_{\gamma\in\Gamma(\mu,\nu)} D^\gamma_\infty(\mu,\nu).
\end{equation}
Actually, $D^\gamma_\infty(\mu,\nu)$ depends only on the support of $\gamma$, i.e., $D^\gamma_\infty(\mu,\nu)=\sup\{D(x,y)\,|\, (x,y)\in\supp(\gamma)\}$.

The existence of an optimal coupling is known, for $1\leq\alpha<\infty$, see \cite[Theorem 4.1]{Villani}, the case $\alpha=\infty$ is shown in \cite[Theorem 2.6]{Jylha2013}, but it does not imply that  $D_\alpha(\mu,\nu)$  is finite.

When $\nu=\delta_y$ is a point measure, there is only one coupling between $\mu$ and $\nu$, namely the product measure $\gamma=\mu\times\delta_y$. Hence \eqref{Wasserstein}/\eqref{Wassersteininfty} reduces to \eqref{DelalphaPoint}/\eqref{DelalphaInfty}. In particular, we indeed get an extension of the given metric for points, interpreted as point measures.
The metric can become infinite, but the triangle inequality still holds \cite{Villani}[after Example~6.3]. The proof relies on Minkowski's inequality and the use of a ``Gluing Lemma'' \cite{Villani}, which builds a coupling from $\mu$ to $\zeta$ out of couplings from $\mu$ to $\nu$ and from $\nu$ to $\zeta$. It also covers the case $\alpha=\infty$, which is not otherwise treated in \cite{Villani}.

Thus the space breaks up into equivalence classes of measures, which have finite distance from each other. In view of the previous section it is natural to consider the class containing all point measures, i.e., the measures of finite spread. However, the restriction to measures of finite spread is not necessary. In fact, for measures on $\Rl$ each equivalence is closed under translations, and the distance of two translates is bounded by the size of the translation.

The metric also has the right {\it scaling}: For $\Omega=\Rl$ let us denote the scaling of measures by $s_\lambda$, so that for $\lambda>0$ and measurable $X\subset\Rl$, $s_\lambda(\mu)(X)=\mu(\lambda^{-1}X)$. Then $D_\alpha(s_\lambda\mu,s_\lambda\nu)=\lambda D_\alpha(\mu,\nu)$, so this metric is compatible with a change of units. Of course, the metric is also unchanged when both measures are shifted by the same translation.

$D_\alpha(\mu,\nu)$ is also known as {\it transport distance}, due to the following interpretation: Suppose that an initial distribution $\mu$ of some ``stuff'' (earth or probability) has to be converted to another distribution $\nu$ by moving the stuff around. The measure $\gamma$ then encodes how much stuff originally at $x$ is moved to $y$. If the transport cost per unit is $D(x,y)^\alpha$, the integral represents the total transport cost. The minimum then is the minimal cost of converting $\mu$ to $\nu$. The root is taken to ensure the right scaling behaviour.

For {\it convexity properties}, note that the function $t\mapsto t^{1/\alpha}$ is concave, so by Jensen's inequality $D_\alpha(\mu,\delta_y)$ is concave in $\mu$, and so is $\Delta_\alpha(\mu)$, as the infimum of concave functions. This is expected, since it is exactly the point measures  that have zero spread, and all other measures are convex combinations of these. The metric is neither concave nor convex in its arguments. However, the function
$\gamma\mapsto D^\gamma_\alpha(\mu,\nu)^\alpha$ is linear, and since $\lambda\gamma_1+(1-\lambda)\gamma_2$ is a coupling between the respective convex combinations of marginals, $(\mu,\nu)\mapsto D_\alpha(\mu,\nu)^\alpha$ is convex. Therefore, the level sets $\{(\mu,\nu)| D_\alpha(\mu,\nu)\leq c\}$ are convex, and $D_\alpha$ is ``pseudoconvex''.

Suppose that $\Omega=\Rl$ and we ``add independent noise'' to a real-valued random variable with distribution $\mu$
by a random translation with distribution $\eta$.  This leads to the convolution   $\eta\ast\mu$ for the new
distribution. The following bounds govern this operation
\begin{lem}Let $\mu,\nu,\eta$ be probability measures on $\Rl$, and $1\leq\alpha\leq\infty$. Then
\begin{eqnarray}\label{convolveBound1}
    \Delta_\alpha(\mu)\leq\Delta_\alpha(\eta\ast\mu)
                             & \leq&\Delta_\alpha(\eta)+\Delta_\alpha(\mu)\\
    D_\alpha(\eta\ast\mu,\eta\ast\nu)&\leq& D_\alpha(\mu,\nu)   \label{convolveBound2}\\
            D_\alpha(\eta\ast\mu,\mu)&\leq& D_\alpha(\eta,\delta_0).\label{convolveBound3}
\end{eqnarray}
\end{lem}
The first inequality says that noise increases spread, but not by more than the spread of the noise. The second says that adding noise washes out the features of two distributions, making them more similar. Finally the third inequality, which will be crucial for us, says that adding a little noise only changes a measure by little. Note that it is not only the spread of the noise which counts here, but also the absolute displacement. That is, a special case of the last relation is that $D_\alpha(\mu^y,\mu)\leq|y|$, where $\mu^y=\delta_y\ast\mu$ is the shift of $\mu$ by $y$. We emphasize that \eqref{convolveBound3} does not require $\Delta_\alpha(\mu)<\infty$.

\begin{proof}
We note that $D_\alpha(\mu,\delta_{y})$ is a standard $p$-norm, or rather $\alpha$-norm $\norm\cdot_{\mu,\alpha}$ in  $\LL^\alpha(\Omega,\mu)$ of the function $x\mapsto(x-y)$. We denote this function as $x-y1$ to indicate that $y$ is considered as a constant. That is
\begin{equation}\label{DeltaLp}
  D_\alpha(\mu,\delta_y)=\norm{x-y1}_{\mu,\alpha}
\end{equation}
This equation is also valid for $\alpha=\infty$, so we need not consider this case separately.

For the first inequality in \eqref{convolveBound1} we use translation invariance and concavity of $\Delta_\alpha$ by considering $\eta\ast\mu$ as a convex combination of translates of $\mu$ with weight $\eta$. For the second inequality in \eqref{convolveBound1} consider the expression
$$\norm{x+y-x'-y'}_{\eta\times\mu,\alpha}=D_\alpha(\mu\ast\nu,\delta_{x'+y'}).$$
This is larger than the infimum over all choices of $(x'+y')$, i.e., $\Delta_\alpha(\mu\ast\nu)$. Using the Minkowski inequality (triangle inequality for the $\alpha$-norm) we conclude
\begin{eqnarray}
  \Delta_\alpha(\mu\ast \nu)
       &\leq&\norm{x-x'+y-y'}_{\eta\times\mu,\alpha}   \nonumber\\
       &\leq&\norm{x-x'}_{\eta\times\mu,\alpha}+\norm{y-y'}_{\eta\times\mu,\alpha}\nonumber\\
       &=&\norm{x-x'}_{\eta,\alpha}+\norm{y-y'}_{\mu,\alpha} \nonumber\\
       &=&D_\alpha(\mu,\delta_{x'})+D_\alpha(\nu,\delta_{y'}) \nonumber
\end{eqnarray}
Here at the last but one equality we used that in the relevant integrals the integrand depends only on one of the two variables $x,y$ and the other is integrated over by a probability measure. The desired inequality \eqref{convolveBound1} now follows by minimizing over $x'$ and $y'$.

Then any coupling $\gamma$ between $\mu$ and $\nu$ provides a coupling $\tilde\gamma(X\times Y)=\int\gamma(X-x,Y-x)\,d\eta(x) $ between
$\eta\ast\mu$ and $\eta\ast\nu$, for which we get $\int|x-y|^\alpha\,d\tilde\gamma(x,y)=\int|x-y|^\alpha\,d\gamma(x,y)$.
Since the infimum may be attained at a coupling different from $\tilde\gamma$ \eqref{convolveBound2} follows.

Finally, to prove \eqref{convolveBound3}, we note that it is a special case of \eqref{convolveBound2} since 
$D_\alpha(\eta\ast\mu,\mu)=D_\alpha(\mu\ast\eta,\mu\ast\delta_0)$.
\end{proof}

A powerful tool for working with the distance functions is a dual expression of the infimum over couplings as a supremum over certain other functions.
A nice interpretation is in terms of transportation prices. We describe it here to motivate the expressions, and refer to the excellent book
\cite{Villani}, from where we took this interpretation, for the mathematical details. In this context we have to exclude the case $\alpha=\infty$.

Suppose the ``stuff'' to be moved is bread going from the bakeries
in town to caf\'es. The transport costs $D(x,y)^\alpha$ per unit, and the bakeries and caf\'es consider hiring a company to take
care of the task. The company will pay a price of $\Psi(x)$ per unit to the bakery at $x$  and charges $\Phi(y)$ from the caf\'e at $y$. Clearly this makes sense if each transport becomes cheaper, i.e.,
\begin{equation}\label{compete}
  \Phi(y)-\Psi(x)\leq D(x,y)^\alpha.
\end{equation}
Pricing schemes satisfying this inequality are called competitive. We are now asking what the maximal gain of a company under a competitive pricing scheme
can be, given the productivity $\mu$ of the bakeries and the demand $\nu$ at the caf\'es. This will be
\begin{equation}\label{gap}
 \int\Phi(y)\,d\nu(y) \ -\  \int\Psi(x)\,d\mu(x) \leq \int D(x,y)^\alpha\,d\gamma(x,y)
\end{equation}
This inequality is trivial from \eqref{compete}, and holds for any pricing scheme $(\Psi,\Phi)$ and any transport plan  $\gamma$. Optimizing the pricing
scheme maximizes the left hand side and optimizing the transport plan minimizes the right hand side. The Kantorovich Duality Theorem asserts that for
these optimal choices the gap closes, and equality holds in \eqref{gap}, i.e.
\begin{equation}\label{DKanto}
  D_\alpha(\mu,\nu)^\alpha=\sup_{\Phi,\Psi} \int\Phi(y)\,d\nu(y) - \int\Psi(x)\,d\mu(x)
\end{equation}
where $\Phi$ and $\Psi$ satisfy \eqref{compete}.

When maximizing the left hand side of \eqref{gap}, one can naturally choose $\Phi$ as large as possible under the constraint \eqref{compete}, i.e.,
$\Phi(y)=\inf_x\{\Psi(x)+D(x,y)^\alpha \}$, and similarly for $\Psi$. Hence one can choose just one variable $\Phi$ or $\Psi$ and determine the other by this
formula. In case $\alpha=1$ the triangle inequality for the metric $D$ entails that one can take $\Phi=\Psi$. In this case \eqref{compete} just asserts that
this function be Lipshitz continuous with respect to the metric $D$, with constant $1$. The left hand side of \eqref{gap} is thus a difference of
expectation values of the given measures $\mu$, $\nu$.

For later purposes we also have to make sure that the duality gap closes if we restrict the set of functions $\Phi,\Psi$. The natural condition is, of course, that $\Psi\in\Ell^1(\mu)$. The statement of Kantorovich Duality in \cite{Villani}[Thm.~5.10] includes that in the supremum \eqref{DKanto} one can restrict to bounded continuous functions. In the same spirit we add

\begin{lem}\label{lem:cptsupp}
In \eqref{DKanto} the supremum can be restricted to positive continuous functions of compact support without changing its value.
\end{lem}

\begin{proof}
Suppose that some bounded continuous functions $\Psi,\Phi$ are given, which satisfy \eqref{compete}.  Since we can add the same constant to each, we may also assume them to be positive. Our aim is to find compactly supported functions $\Psi_\eps,\Phi_\eps$ such that $0\leq\Psi_\eps\leq\Psi$, $\int(\Psi(x)-\Psi_\eps(x))d\mu(x)\leq\eps$, and similarly for $\Phi_\eps$. The problem is to find such functions so that \eqref{compete} still holds.

Pick a compact region $U$ so that $\int_{y\notin U}\Phi(y)\,d\nu(y)<\eps$, and some continuous function $0\leq\Phi_\eps\leq\Phi$ coinciding with $\Phi$ on $U$ and vanishing outside a compact set $\widehat U\supset U$. Let $V$ be a set on which $\Psi d\mu$ similarly achieves its integral to within $\veps$, and which also contains $\widehat U$ and all points of distance at most $\norm{\Phi}_\infty^{1/\alpha}$ from it. Construct a compactly supported function $\Psi_\veps$ coinciding with $\Psi$ on $V$. Consider now the inequality $$\Phi_\veps(y)-\Psi_\veps(x)\leq D(x,y)^\alpha.$$ Clearly this holds for all $x\in V$, because then $\Psi_\eps(x)=\Psi(x)$ and $\Phi_\eps(y)\leq\Phi(y)$. For $y\in\widehat U$ and $x\notin V$, we have
$$D(x,y)^\alpha-\Phi_\eps(y)\geq D(x,y)^\alpha-\Phi(y)\geq D(x,y)^\alpha-\norm\Phi_\infty\geq0.$$
Hence the inequality follows from $\Psi_\eps\geq0$. Finally for $y\notin\widehat U$ we have $\Phi_\veps(y)=0$, and the inequality is once again trivial.
To summarize, we have shown it on $(\Omega\times V)\cup(\widehat U\times V^c)\cup(\widehat U^c\times\Omega)=\Omega\times\Omega$.

Of the properties of $\Omega=\Rl$ we only used (for the compactness of $V$) that closed balls of the metric are compact.
\end{proof}

\begin{example}\label{ex:Wass2}$2$-distance and affine families.\\ \rm
It is instructive to see just how far one can go by taking $\Phi,\Psi$ in \eqref{gap} to be quadratic expressions in the case $\alpha=2$.
So let $\Psi(x)=(a-1)x^2+2bx$ with $a>0$.  Then
\begin{equation}\label{phiquad}
  \Phi(y)=\inf_x\{\Psi(x)+(x-y)^2 \}=\bigl(1-\frac1a\bigr)y^2+\frac{2by}a-\frac{b^2}a
\end{equation}
Now let $\mu,\nu$ be probability measures with finite second moments, say means $m(\mu)=\mu[x],m(\nu)=\nu[x]$ and variances $s(\mu)^2,s(\nu)^2$. Then the left hand side of \eqref{gap} can be entirely expressed by these moments, and we get a lower bound
\begin{eqnarray}\label{Wass2Moments}
  D_2(\mu,\nu)^2
    &\geq& (s(\mu)-s(\nu))^2+(m(\mu)-m(\nu))^2\\
    &\geq&
     \frac{a-1}a\,(s(\mu)^2+m(\mu))^2+\frac{2b}a\,m(\mu)-\frac{b^2}a -(a-1) (s(\nu)^2+m(\nu)^2)- b\, m(\nu) \nonumber
\end{eqnarray}
Here the second expression is what one gets by just inserting the moments (e.g., $m(\mu)$ and $s(\mu)^2+m(\mu)^2$) into \eqref{gap}, and the first is the result of maximizing over $b$ and $a>0$.
Turning to the upper bound, it is not a surprise that the maximization for getting $\Psi$ leaves a quadratic expression for the difference
\begin{equation}\label{Wass2diff}
  (x-y)^2+\Psi(x)-\Phi(y)=\frac1a(ax+b-y)^2\ .
\end{equation}
Hence for equality in \eqref{gap}, and therefore in both inequalities of \eqref{Wass2Moments} to hold we must have a coupling
$\gamma$ which is concentrated on the line $y=ax+b$, so $\int g(x,y)\gamma(dx\,dy)=\int g(x,ax+b)d\mu(x)$ for any test function $g$. In particular, $\nu$ must arise from $\mu$ by translation and dilatation. But if that is the case the moments just have to come out right, so the converse is also true. Hence we have a very simple formula for $D_2$ on any orbit of the affine group, for example the set of Gaussian probability distributions. The argument also allows us to find the shortest $D_2$-distance from a measure $\mu$ to a set of measures with fixed first and second moments: The closest point will be the appropriate affine argument transformation applied to
$\mu$. This, with much further information about $D_2$ geodesics and the connection with Legendre transforms is to be found in \cite[Thm.~3.1]{Carlen}.

With the above assumption of finite second moments we also get an upper bound, so that
$$
(s(\mu)-s(\nu))^2+(m(\mu)-m(\nu))^2\leq D_2(\mu,\nu)^2\leq(s(\mu)+s(\nu))^2+(m(\mu)-m(\nu))^2,
$$
with the bounds obtained if there is a $\gamma$ giving strong negative, resp. positive correlation for $\mu$ and $\nu$,
making them linearly dependent.
\end{example}

\subsection{Metric distance of observables}\label{sec:distobs}
Given an $\alpha$-deviation for probability distributions we can directly define an $\alpha$-deviation for  observables $E,F$ with the same metric outcome space
$(\Omega,D)$. We set, for $1\le\alpha\le\infty$,
\begin{equation}\label{distobs}
  D_\alpha(F,E):=\sup_\rho D_\alpha(F_\rho,E_\rho).
\end{equation}

Note that we are taking here the {\it worst case} with respect to the input states. Indeed we consider the deviation of an observable $F$ from an
``ideal''  reference observable $E$ as a figure of merit for $F$, which a company might advertise: No matter what the input state, the distribution
obtained by $F$ will be $\veps$-close to what you would get with $E$. When closeness of distributions is measured by $D_\alpha$, then \eqref{distobs} is the best $\veps$ for which this is true. In contrast, the individual deviations $D_\alpha(F_\rho,E_\rho)$ are practically useless as a benchmark. Indeed, a testing lab, which is known to always use the same input state for its tests, is easily fooled. Their benchmark could be met by any fake device, which does not make any measurement, but instead produces random numbers with the expected distribution. Put in colloquial terms: Nobody would buy a meter stick which is advertised as ``very precise, provided the length measured is 50 cm'', or a watch which ``shows the correct time twice a day''.

The additional maximization in \eqref{distobs} leads to some simplifications, and in particular to an explicit expression for the difference between a sharp observable and the same observable with added noise.

\begin{lem}\label{lem:DobsConv}
Let $E$ be a sharp observable on $\Rl$, $\eta$ some probability measure on $\Rl$, and $F=\eta\ast E$, i.e., $F_\rho=\eta\ast E_\rho$ for all $\rho$. Then
\begin{equation}\label{DobsConv}
  D_\alpha(F,E)= D_\alpha(\eta,\delta_0)\ .
\end{equation}
\end{lem}

\begin{proof}
By \eqref{convolveBound3} we have  $D_\alpha(F_\rho,E_\rho)\leq D_\alpha(\eta,\delta_0)$. We claim that this upper bound is nearly attained whenever $E_\rho$ is sharply concentrated, say, $D_\alpha(E_\rho,\delta_q)\leq\veps$; this is possible, because $E$ was assumed to be sharp. Indeed we then have
$D_\alpha(\eta,\delta_0)=D_\alpha(\eta\ast\delta_q,\delta_q)
   \leq D_\alpha(\eta\ast\delta_q,\eta\ast E_\rho)+D_\alpha(\eta\ast E_\rho,E_\rho)+D_\alpha(E_\rho,\delta_q)
   \leq 2\veps+ D_\alpha(\eta\ast E_\rho,E_\rho) = 2\veps+D_\alpha(F_\rho,E_\rho)$.
\end{proof}

\begin{example}\label{standardmodel}\rm
The standard model for measuring  (approximately) a sharp observable associated with the selfadjoint operator
$A$ consists in coupling the system with a probe, with the Hilbert space ${\mathcal L}^2(\Rl)$, using the direct interaction $e^{i\lambda A\otimes P_p}$, and monitoring the shifts in the probe's position $Q_p$.
If the probe is initially prepared in a state $\sigma$, then the actually measured observable $F$ is a smearing of $E^A$, with the ($\lambda$-scaled) probability density of the  probe position in state
$\Pi\sigma\Pi$.
Thus we get   $D_\alpha(F,E^A)=D_\alpha(E^{Q_p/\lambda }_{\Pi\sigma\Pi},\delta_0)$.
This shows that the error in measuring $E^A$ with the standard model can be made arbitrarily small with an appropriate choice of the initial probe state $\sigma$ but can never be made equal to 0.
\end{example}

\begin{example}\label{smcontinues}\rm
If the standard measurement of a sharp observable $A$ is followed immediately (in the sense that any free evolution in between can be neglected) by a measurement of another sharp observable $B$, then the resulting (sequential) joint measurement constitutes an approximate joint measurement of $A$ and $B$, with the first marginal observable $M_1$ being a smearing of $A$, as given in Example \ref{standardmodel}, and the second marginal $M_2$ is a distorted version of $B$,
$$
M_2(Y)=\mathcal I(\Rl)^*(E^B(Y)) = \int_{\Rl} K_x^*E^B(Y)K_x\,dx,
$$
where $K_x=\int \sqrt\lambda\phi(-\lambda(y-x))\,dE^A(y) = \sqrt\lambda \phi(-\lambda(A-x))$ for all $x\in\Rl$; for simplicity, we have assumed here that the initial probe state $\sigma$ is a pure state given by a function $\phi\in L^2(\Rl)$ of unit norm.

While $D_\alpha(M_1,E^A)$ is easily computed, the error $D_\beta(M_2,E^B)$ can be determined only if $A$ and  $B$ are explicitly given.

For instance, if $A=Q$ and $B=P$,
then $M_2$ is a smearing of $E^P$, with the convolving probability measure being the ($1/\lambda$-scaled) momentum distribution of the probe in the state $\Pi\sigma\Pi$.
A standard position measurement followed by a momentum measurement turns out to be an implementation of a covariant phase space observable $M^\tau$, $\tau$ depending on $\sigma$.
In this case,
 $D_\alpha(M_1,E^Q)D_\beta(M_2,E^P)=D_\alpha(\mu_\tau,\delta_0) D_\beta(\nu_\tau,\delta_0)$, which reduces to the generalized Kennard-Robertson-Weyl inequality of Proposition \ref{thm:prepURpq}.
For $\alpha=\beta=2$ one thus gets
$D_2(M_1,E^Q)D_2(M_2,E^P)\geq \hbar/2$, where the lower bound is reached
exactly when $\tau$ is a
centered minimal uncertainty state, that is, $\tau$ is a pure state given by a real valued Gaussian wave function $\phi$ whose position and momentum distributions are
centered at the origin.

For $\alpha=\infty$, say, the finiteness of the uncertainty product implies that $\beta<\infty$ since there is no $\tau$ for which both the position and momentum distributions would have  bounded supports.
\end{example}

\begin{remark}\rm
As seen above,
any covariant phase space observable $M^\sigma$ can be realized, for instance, as a standard (approximate) position measurement followed by a momentum measurement, the generating operator $\sigma$ depending on the initial probe state. A more realistic implementation of an $M^\sigma$ can be obtained as the high amplitude limit of the signal observable measured by an eight-port homodyne detector; for details, see \cite{KL2008b}.
\end{remark}

\subsection{Calibration error}

The idea of looking especially at states for which $E_\rho$ is sharply concentrated can be used also in a more general setting, and even gives a
possible alternative definition of the error quantities. The idea is that the supremum \eqref{distobs} over all states is not easily accessible to
experimental implementation. It seems more reasonable to just calibrate the performance of a measurement   $F$ as an approximate measurement of $E$ by looking
at the distributions $F_\rho$ for preparations for which $E_\rho$ is nearly a point measure, i.e., those for which $E$ ``has a sharp value''. The idea of calibration error was formalized in \cite{BuPe07} as a measure of {\em error bar width} which was shown to
obey a measurement uncertainty relation using the method developed in \cite{Werner04} and applied here.

For $\veps>0$, we define the {\it $\veps$-calibration error}, resp. the {\it calibration error} of $F$ with respect to $E$ as
\begin{eqnarray}\label{calbratEps}
  \Delta_\alpha^\veps(F,E)&=&\sup_{\rho,x}\Bigl\{D_\alpha(F_\rho,\delta_x)\Bigm| \ D_\alpha(E_\rho,\delta_x)\leq\veps\Bigr\}\\
  \Delta_\alpha^c(F,E)    &=&\lim_{\veps\to0}\Delta_\alpha^\veps(F,E)  \label{calbrate}
\end{eqnarray}
Here the limit in \eqref{calbrate} exists because \eqref{calbratEps} is a monotonely decreasing function.
By the triangle inequality, we have  $D_\alpha(F_\rho,\delta_x)\leq D_\alpha(F_\rho,E_\rho)+D(E_\rho,\delta_x)$ and, taking the supremum over $\rho$ and $x$ as in \eqref{calbratEps}
\begin{equation}\label{cal<Del}
  \Delta_\alpha^\veps(F,E)\leq D_\alpha(F,E)+\veps \quad\mbox{and }\ \Delta_\alpha^c(F,E)\leq D_\alpha(F,E).
\end{equation}
When $F$ just adds independent noise, there is also the complementary inequality, the direct analog of Lemma~\ref{lem:DobsConv}.

\begin{lem}\label{lem:DcobsConv}
Let $E$ be a sharp observable on $\Rl$, $\eta$ some probability measure on $\Rl$, and $F=\eta\ast E$. Then
\begin{equation}\label{DcobsConv}
      D_\alpha(\eta,\delta_0)-\veps \leq  \Delta^\veps_\alpha(F,E)\leq D_\alpha(\eta,\delta_0)+\veps\ ,
\end{equation}
so that letting $\veps\to0$ yields $\Delta^c_\alpha(F,E)=D_\alpha(\eta,\delta_0)$.
\end{lem}

\begin{proof}
For any calibration state $\rho$, i.e., $D_\alpha(E_\rho,\delta_x)\leq\veps$, we have the upper bound
$D_\alpha(F_\rho,\delta_x)=D_\alpha(\eta\ast E_\rho,\delta_x)
   \leq D_\alpha(\eta\ast E_\rho,E_\rho)+D_\alpha(E_\rho,\delta_x)\leq D_\alpha(\eta,\delta_0)+\veps$.
For the complementary bound we use
$D_\alpha(\eta,\delta_0)=D_\alpha(\eta\ast\delta_x,\delta_x)
  \leq D_\alpha(\eta\ast\delta_x,\eta\ast E_\rho)+D_\alpha(\eta\ast E_\rho,\delta_x)\le \veps+D_\alpha(\eta\ast E_\rho,\delta_x)$.
Hence
\begin{eqnarray}
       D_\alpha(\eta,\delta_0)- \veps&\leq D_\alpha(F_\rho,\delta_x)\leq& D_\alpha(\eta,\delta_0)+\veps  \nonumber\\
       D_\alpha(\eta,\delta_0) -\veps                 &\leq \Delta^\veps_\alpha(F,E)\leq& D_\alpha(\eta,\delta_0)+\veps
\label{calibuplow}
\end{eqnarray}
where the second row is the supremum of the first over all $x$ and all calibrating states.
\end{proof}

Hence, in the case of convolution observables $F=\eta\ast E$ we have $D_\alpha(F,E)=\Delta^c_\alpha(F,E)$. In general, however, the inequality \eqref{cal<Del} is strict, as is readily seen by choosing a discrete metric on two points ($\Omega=\{0,1\}$). Then $D_\alpha(F,E)=\sup_\rho\,|\tr{\rho(F(\{1\})-E(\{1\}))}|/2$, but $\Delta_\alpha^c(F,E)$ is a similar expression with $\rho$ constrained to diagonal pure states.

\section{Error-disturbance relations: position and momentum}\label{sec:QPUR}
\subsection{Preparation uncertainty}
Before we can formulate the main result of our paper we state a generalization of the Kennard-Robertson-Weyl inequality for the $\alpha$-spreads introduced in \eqref{Delalpha}. This result improves the inequality derived by Hirschman \cite{Hirschman},
see also \cite{Cowling,Folland}, in that we now have an optimal lower bound. The details of the constants $c\sab$ as a function of $\alpha$ and $\beta$ will be studied numerically in Sect.~\ref{sec:constants}, with an overview given in Fig.~\ref{fig:ccprime}.

\begin{proposition}[Preparation Uncertainty]\label{thm:prepURpq}
Let $E^Q$ and $E^P$ be canonically conjugate position and momentum observables, and $\rho$ a density operator. Then, for any $1\leq \alpha,\beta<\infty$,
\begin{equation}\label{prepURpq}
  \Delta_\alpha(E_\rho^Q)\Delta_\beta(E_\rho^P)\geq c\sab\hbar,
\end{equation}
The constant $c\sab$ is connected to the ground state energy $g\sab$ of the Hamiltonian $H\sab=|Q|^\alpha+|P|^\beta$ by the equation
\begin{equation}\label{gcpq}
    c\sab = \alpha^{\frac 1\beta}\beta^{\frac 1\alpha}\left(\frac{g\sab}{\alpha+\beta}\right)^{\frac 1\alpha+\frac 1\beta}.
\end{equation}
The lower bound is attained exactly when $\rho$ arises from the ground state of the operator $H\sab$ by phase space translation and dilatation.
For $\alpha=\beta=2$, $H\sab$ is twice the harmonic oscillator Hamiltonian with ground state energy $g_{22}=1$, and  $c_{22}=1/2$.
\end{proposition}

The Hamiltonian $H\sab$ appears here mainly through the quadratic from $\brAAket\psi{H\sab}\psi$ where $\psi$ runs over, say, the unit vectors in the Schwartz space of tempered functions. The inequality \eqref{prepURpq} depends only on the lower bound $g\sab$ of this form.

This makes sense also for $\alpha=\infty$, when $\brAAket\psi{\abs q^\infty}\psi$ is interpreted by the limit $\alpha\to\infty$, i.e., as $\infty$ unless
$\psi$ vanishes almost everywhere outside $[-1,1]$, in which case the expectation is zero. The effect of this singular ``potential'' is to confine the particle to the box $[-1,1]$. Note that in this case \eqref{gcpq} simplifies to $c_{\infty\beta}=g_{\infty\beta}^{1/\beta}$. Of course, since $\psi$ cannot be compactly supported in both position and momentum we have $c_{\infty\infty}=\infty$.

\begin{proof}
Consider the family of  operators
\begin{equation}\label{Hscaled}
  H\sab(p,q,\lambda)=\lambda^\alpha \abs{Q-q\idty}^\alpha + \lambda^{-\beta}\abs{P-p\idty}^\beta\
       \geq g\sab\,\idty,
\end{equation}
obtained from $H\sab$ by shifting in phase space by $(q,p)$, and by a dilatation $(Q,P)\mapsto(\lambda Q,\lambda^{-1}P)$. Since these operations are unitarily implemented, the lower bound $g\sab$ for all these operators is  independent of $p,q,\lambda$. Now, for a given $\rho$, we may assume that $\Delta_\alpha(E_\rho^Q)$ and $\Delta_\beta(E_\rho^P)$ are both finite, since these uncertainties do not vanish for any density operator, and one infinite factor hence renders the inequality trivial. Let $q$ be the point for which $D_\alpha(E_\rho^Q,\delta_q)$ attains its minimum $\Delta_\alpha(E_\rho^Q)$, and choose $p$ similarly for $P$. Then by taking the expectation of \eqref{Hscaled} with $\rho$ we obtain the inequality
\begin{equation}\label{rhoHscaled}
  \lambda^\alpha\Delta_\alpha(E_\rho^Q)^\alpha+\lambda^{-\beta}\Delta_\beta(E_\rho^P)^\beta\geq g\sab.
\end{equation}
The minimum of the left hand side with respect to $\lambda$ is attained at
\begin{equation}\label{lambdaMin}
  \lambda=\left(\frac{\beta\, \Delta_\beta(E_\rho^P)^\beta}{\alpha\,\Delta_\alpha(E_\rho^Q)^\alpha }\right)^{1/({\alpha+\beta})}.
\end{equation}
Inserting this into \eqref{rhoHscaled} gives an expression that depends only on the uncertainty product $u=\Delta_\alpha(E_\rho^Q)\Delta_\beta(E_\rho^P)$, namely
\begin{equation}\label{uproductbound}
  u^{\alpha\beta/(\alpha+\beta)}\,\alpha^{-\alpha/(\alpha+\beta)}\,\beta^{-\beta/(\alpha+\beta)}\,(\alpha+\beta)\geq g\sab.
\end{equation}
Now solving for $u$ gives the uncertainty inequality. Moreover, since the left hand side is still nothing but the expectation of $H\sab(p,q,\lambda)$ with a suitable choice of parameters, equality holds exactly if $\rho$ is the ground state density operator of $H\sab(p,q,\lambda)$. But since this operator arises by dilatation and shifts from $H\sab$, its ground state must arise by the same operations from the ground state of $H\sab$.
\end{proof}

For the statements about equality and near equality, which are the subject of the following theorem, we need more information about the operator $H\sab$, particularly its low-lying eigenvalues. Thus we have to turn the quadratic form into a bona fide selfadjoint operator by the Friedrichs extension.
This approach also regulates how to understand the case $\alpha=\infty$; the resulting operator lives on $\LL^2([-1,1],dq)$, with the domain chosen so that the extension of the function to the whole line is in the domain of $\abs P^\alpha$. This requires the function and some derivatives to vanish at the boundary. Are there eigenvalues at the bottom of the spectrum? Intuitively, $H\sab$ is the quantization of a phase space function diverging at infinity, so should have a purely discrete spectrum with eigenvalues going to infinity. This can be verified by  the Golden-Thompson inequality according to which, for any $\lambda\geq0$,
\begin{equation}
\tr{e^{-\lambda H\sab}}
    \leq \tr{e^{-\lambda|Q|^\alpha}e^{-\lambda|P|^\beta}}
    =\int e^{-\lambda|q|^\alpha}\,dq\int e^{-\lambda|p|^\beta}\,dp <\infty,
\end{equation}
see, for instance, \cite[p. 94]{Simon1979}. Thus, the positive operator on the left is trace class (and thus compact), and,
therefore, the spectrum of the generator  of the semigroup $e^{-\lambda H\sab}$,
$\lambda>0$, consists of a countable discrete set of eigenvalues each of finite multiplicity \cite[Theorem 2.20]{Davies1980}.
Since each of the terms in $H\sab$ already has strictly positive expectation in any state, the lowest eigenvalue $g\sab$ is strictly positive.
Although this might be an interesting task, we do not prove more fine points about the spectrum of $H\sab$ in this paper. Supported by the numerical calculations (on which we anyhow have to rely for the concrete values) and some solvable cases, we take for granted that the ground state is non-degenerate and lies in the symmetric subspace, and the first excited state has strictly higher energy and lies in the odd subspace.
The gap $g'\sab-g\sab$ plays a crucial role in showing the stability of the minimizing states: A state with near-minimal uncertainty product must be close to a state with exactly minimal uncertainty product. The precise statement is as follows:

\begin{proposition}[Near-minimal Preparation Uncertainty]\label{thm:NprepURpq}
Under the conditions of Proposition~\ref{thm:prepURpq} consider the case of equality  in \eqref{prepURpq}.
Suppose that this is only nearly the case, i.e., the uncertainty product is
\begin{equation}\label{NprepURpq}
 c\sab\hbar\le u=\Delta_\alpha(E_\rho^Q)\Delta_\beta(E_\rho^P)<  c'\sab\hbar,
\end{equation}
where $c'\sab$ is related to the energy $g'\sab$ of the first excited state of $H\sab$ by \eqref{gcpq}. Then there is a state $\rho'$ minimizing \eqref{prepURpq} exactly, such that with $\gamma=\alpha\beta/(\alpha+\beta)$,
\begin{equation}\label{nearPURbound}
   \norm{\rho-\rho'}_1\leq 2\sqrt{\frac{u^\gamma-c\sab^\gamma}{c\sab^{\prime\,\gamma}-c\sab^\gamma}}\,.
\end{equation}
\end{proposition}
The bound in  \eqref{nearPURbound} approaches zero as $u\to c\sab$, and becomes vacuous for $u>c'\sab$. The constants $c\sab,c'\sab$ are shown in Fig.~\ref{fig:ccprime}, and indications how to compute them will be given in Sect.~\ref{sec:constants}.

The assumption \eqref{NprepURpq} entails that  the left hand side of \eqref{uproductbound} is below the next eigenvalue $g'\sab$. In this case we get information about how close $\rho$ must be to the ground state $\rho'$. This is obtained via a simple and well-known Lemma, whose straightforward application to \eqref{uproductbound} then gives the inequality \eqref{nearPURbound}.

\begin{lem}Let $H$ be a selfadjoint operator with non-degenerate ground state $\psi_0$ with energy $E_0$ such that the rest of the spectrum lies above
$E_1>E_0$. Then for any density operator $\rho$ we have $\braket{\psi_0}{\rho\psi_0}\geq (E_1-\tr\rho H)/(E_1-E_0)$ and the trace norm bound
\begin{equation}\label{gaplemma}
  \Bigl\Vert \rho-\kettbra{\psi_0}\Bigr\Vert_1 \leq 2\sqrt{\frac{\tr\rho H-E_0}{E_1-E_0}}.
\end{equation}
\end{lem}
\begin{proof}
The statements about the spectrum of $H$ are equivalent to the operator inequality
\begin{equation}\label{opgap}
  H\geq E_0 \kettbra{\psi_0}+ E_1\bigl(\idty-\kettbra{\psi_0}\bigr).
\end{equation}
Taking the trace with $\rho$ gives the bound on the fidelity $f=\braket{\psi_0}{\rho\psi_0}$. The bound $\norm{\rho-\kettbra{\psi_0}}_1\leq2\sqrt{1-f}$ holds in general, and is proved easily for pure states $\rho$ and extended to mixed ones by Jensen's inequality for the concave square root function.
\end{proof}

\begin{figure}[ht]
\centering
  \includegraphics[width=12cm]{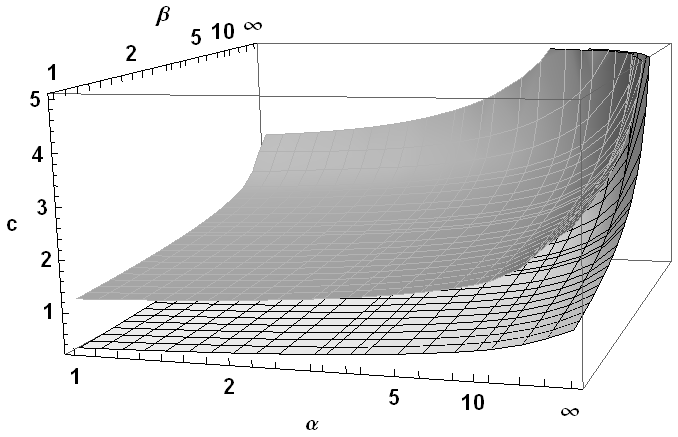}
	\caption{The constants $c\sab$ and $c'\sab$ appearing in Propositions~\ref{thm:NprepURpq} and \ref{thm:prepURpq}, as determined numerically in
             Sect.~\ref{sec:constants}. The $\alpha$ and $\beta$ axes have been scaled according to $\alpha\mapsto(\alpha-1)/(\alpha+1)$ to include the whole range $1\leq\alpha\leq\infty$.}
	\label{fig:ccprime}
\end{figure}

\subsection{The covariant case}
Here we will study the simple case of covariant observables described in Sect.~\ref{sec:covar}. By virtue of \eqref{povm} these are parameterized by a density operator $\sigma$, and we will see that both measurement uncertainty quantities (by distance of observables and by calibration) for the marginals of the covariant observable are simply equal to the corresponding preparation uncertainties for the density operator $\sigma$, considered as a state. This explains why the constants for our measurement and preparation uncertainties are also just the same. The formal statement is as follows:

\begin{proposition}\label{thm:MURcov}
Let $M$ be a covariant phase space observable, generated by a density operator $\sigma$, with position and momentum marginals $M^Q$ and $M^P$. Then for all $\alpha,\beta\in[1,\infty]$:
\begin{eqnarray}
  D_\alpha(M^Q,E^Q)=\Delta^c_\alpha(M^Q,E^Q)&=&D_\alpha(E^Q_\sigma,\delta_0)\, , \\
  D_\beta(M^P,E^P)=\Delta^c_\beta(M^P,E^P)&=&D_\beta(E^P_\sigma,\delta_0)\ .
\end{eqnarray}
Suppose that the product $u$ of these uncertainties is close to its minimum $c\sab$. Then there is another covariant observable $M'$ with exactly minimal  uncertainty product such that
\begin{equation}\label{nearMURbound}
  \norm{M-M'}\leq2\sqrt{\frac{u^\gamma-c\sab^\gamma}{c\sab^{\prime\,\gamma}-c\sab^\gamma}}
\end{equation}
\end{proposition}

The finiteness of the uncertainty product implies
 that $E^Q_\sigma,E^P_\sigma$ have finite moments of degree $\alpha,\beta<\infty$.
If $\alpha=\infty$  then $\supp(E^Q_\sigma)$ is bounded, $\beta<\infty$ and $\int |p|^\beta\,dE^P_\sigma(p)<\infty$.

\begin{proof}
The equalities are direct applications of Lemmas \ref{lem:DobsConv} and \ref{lem:DcobsConv}.  Therefore, for a near-minimal uncertainty product, we can apply Prop.~\ref{thm:NprepURpq} to conclude that there is a density operator $\sigma'$ with exactly minimal uncertainty product, which is norm close to $\sigma$. Then the corresponding covariant observable is also close to $M$. It remains to show the norm estimate $\norm{M-M'}\leq\norm{\sigma-\sigma'}_1$. This follows because, for any input state $\rho$, we have $\norm{M_\rho-M'_\rho}_1=\norm{\rho\ast(\sigma-\sigma')}_1\leq\norm{\sigma-\sigma'}_1$.
\end{proof}

\subsection{The general case}
The main result of this paper is the following measurement uncertainty relation.

\begin{thm}\label{thm:main}
Let $M$ be a  phase space observable and $1\leq\alpha,\beta\leq\infty$. Then
\begin{eqnarray}
  D_\alpha(M^Q,E^Q)\, D_\beta(M^P,E^P)\               &\geq& c\sab\hbar\quad \mbox{and}\\
  \Delta^c_\alpha(M^Q,E^Q)\, \Delta^c_\beta(M^P,E^P)\ &\geq& c\sab\hbar\, ,
\end{eqnarray}
provided that in each inequality the quantities on the left hand side are finite. The constants $c\sab$ are the same as in Proposition~\ref{thm:prepURpq}.
\end{thm}

Note that the proviso rules out the indefinite product $0\cdot\infty$, along with the utterly uninteresting case $\infty\cdot\infty$. Examples for the indefinite case can be given quite easily. It suffices to combine an ideal position measurement with a random momentum output. Although the statement given here seemingly excludes the indefinite case, it is actually the best one can say about it: If the uncertainty relation is to express quantitatively that not both  $D_\alpha(M^Q,E^Q)$ and $D_\beta(M^P,E^P)$ can become small, then we should also have that if one is zero, the other must be infinite. But this statement is implied by the Theorem, which shows that the case $0\cdot$finite does not occur. Of course, we can also conclude that if in some process one uncertainty tends to zero the other has to diverge in keeping with the Theorem. That is, the indefinite case as approached from less idealized situations is also covered and interpreted as ``$0\cdot\infty\geq c\sab$''.

The reason this indefinite case does not occur for preparation uncertainty is that we have restricted ourselves to states given by density operators, for which $\Delta_\alpha(E_\rho^{P,Q})$ cannot vanish. Among the so-called singular states (positive normalized expectation value functional on the bounded operators which are not given by density operators) one also finds examples of the indefinite case. Singular states with sharp position assign zero probability to every finite interval of momenta. The momentum is thus literally infinite with probability one, not just a distribution on $\Rl$ with diverging moments. An example from the literature is the (mathematically cleaned up version of) the state used in the paper of Einstein, Podolsky and Rosen, where the difference of the two position observables is supposed to be sharp, and accordingly the conjugate difference of momenta is infinite. This also implies that all measurement outcomes seen by Alice or Bob for $Q_1,Q_2,P_1,P_2$ are infinite with probability one. A detailed study is given in \cite{KSW03}. It is all not as strange as it may seem, as one can see if one replaces the EPR state by a highly (but not infinitely) squeezed  two-mode Gaussian state. It is then clear that all individual measurements $Q_1,Q_2,P_1,P_2$ have very broad distributions, and in the limit the probability for any finite interval goes to zero.

The proof will use the Kantorovich dual characterization \eqref{DKanto} of Wasserstein metrics, and thereby excludes the case of one infinite exponent. However, both sides of the inequality are continuous at $\alpha\to\infty$, $\beta$ fixed, so it actually suffices to consider $\alpha,\beta<\infty$, which we will do from now on. The proof of the Theorem is by reduction to the covariant case, i.e., Proposition~\ref{thm:MURcov} combined with Proposition~\ref{thm:prepURpq}. The following Proposition summarizes what we need.

\begin{proposition}\label{crucial-lemma}
Let $M$ be a  phase space observable and $1\leq\alpha,\beta<\infty$. Suppose that $D_\alpha(M^Q,E^Q)$ and  $D_\beta(M^P,E^P)$ are both finite. Then there is a covariant observable $\Mav$ such that
$$D_\alpha(\Mav^Q,E^Q)\leq D_\alpha(M^Q,E^Q)\quad\mbox{and }\quad D_\beta(\Mav^P,E^P)\leq D_\beta(M^P,E^P).$$
The analogous statement holds for calibration measures $\Delta_\alpha^c$ instead of metric distances $D_\alpha$.
\end{proposition}

The basic technique for the proof is averaging over larger and larger sets in phase space, and a compactness argument, that asserts that such an averaging process will have a limit. The most convenient general result based on just this idea is the Markov-Kakutani Fixed Point Theorem \cite[Thm.~V.10.6]{Dunford}. It says that a family of commuting continuous affine isomorphisms of a compact convex set must have a common fixed point. In our case the set in question will be the set of observables with given finite measurement uncertainties, and the transformations are the phase space translations $M\mapsto M^\qp$ defined as
\begin{equation}\label{Mshift}
  M^\qp(Z)=W\qp^* M\bigl(Z+\qp\bigr)W\qp \,
\end{equation}
for any measurable set $Z$ in phase space. Note that this combines a Weyl translation with a translation of the argument in such a way that the common fixed points of these translations are precisely the covariant observables.

In order to satisfy the premises of the Markov-Kakutani Theorem we have to define a topology for which the phase space translations are continuous, and  for which the sets of observables with fixed finite uncertainties are compact. As often in compactness arguments this is the only subtle point, and we will be didactically explicit about it. The topology will be the ``weak'' topology, i.e., the  ``initial topology'' \cite[Sect.~I.\S2.3]{BourbakiTop} which makes the functionals
\begin{equation}\label{initialtopology}
  M\mapsto u_M(\rho,f)=\int f\qp\ dM_\rho\qp    \quad\mbox{for}\ \rho\in\trcl(\HH),\ f\in\Co,
\end{equation}
continuous. That is, the neighbourhoods are specified by requiring a finite number of these functionals to lie in an open set. Let $I$ denote the set of pairs
$(\rho,f)$ of a density operator $\rho$ on $\HH$ and a function $f\in\Co$ with $0\leq f\leq1$. Then for each such pair and every observable $M$ we have $u_M(\rho,f)\in[0,1]$, which we consider as the $(\rho,f)$-coordinate of a point $M^\Box$ in the cube $\cube$. It is clear that $M^\Box$ determines $M$ uniquely: the functional $\rho\mapsto u_M(\rho,f)$ is affine (i.e., respects convex combinations), so there is a unique operator $M'(f)$, with $u_M(\rho,f)=\tr\rho M'(f)$. Since $f\mapsto M'(f)$ is also affine, we can reconstruct the measure $M$ from it so that $M'(f)=\int\!f(x)\,dM(x)$. We can therefore look at the observables as a subset of $\cube$. By definition, the weak topology on the set of observables is the one inherited from the product topology on the cube. By Tychonov's Theorem this is a compact set. Hence this theorem, which embodies the Axiom of Choice, will be the source for all compactness statements about observables in the sequel.

From the proof that $M\mapsto M^\Box$ is injective it is clear that most points in $\cube$ do not correspond to observables. This suggests to single out the subset $\cubo\subset\cube$ of points which are affine in both $\rho$ and $f$. Note that an affinity condition like $\lambda u_M(\rho_1,f)+(1-\lambda)u_M(\rho_2,f)-u_M(\lambda\rho_1+(1-\lambda)\rho_2,f)=0$ involves only three coordinates at a time. Therefore, the left hand side of this equation is continuous, and the subset on which it is true is closed as the inverse image of $\{0\}$ under a continuous function. Since the arbitrary intersection of closed sets is closed, we conclude that $\cubo$ is compact, as a closed subset of a compact set. A similar argument shows that the Weyl translations are continuous on $\cube$. Indeed, to make any finite number of coordinates of a Weyl-translate $M^\qp$ lie in a specified open set, we only need to shift every $\rho$ and $f$ in this neighborhood description to find an appropriate condition on $M$.

However, $\cubo$ is not exactly the set of observables, because it also contains the zero element. What we get from an arbitrary point $M^\Box\in\cubo$ is an operator valued measure $M$, which however need not be normalized. The subset of normalized observables, i.e., those which formally satisfy $u_M(\rho,1)=1$ is {\it not} closed, simply because  $1\notin\Co$. One can define the normalization operator for every $M^\Box\in\cubo$ as
\begin{equation}\label{normM}
  M(1)=\sup_{f\leq1}\int\!f\qp\ dM\qp
\end{equation}
since the net of functions $f\in\Co$ is directed, and so in the weak operator topology the limit of any increasing net in $\Co$, which pointwise goes to $1$ is the same. However, this limit is not a continuous function in the weak topology on observables, and may well be strictly smaller than $\idty$. In fact, it is easy to construct sequences of normalized observables which converge to zero: it is enough to shift any observable to infinity, i.e., to take $N_\qp(Z)=M(Z+\qp)$ without the Weyl operators used in \eqref{Mshift}. The region where the probability measure $\tr\rho M(\cdot)$ is mostly concentrated will thus be shifted away from the region where a function $f\in\Co$ is appreciably different from zero, with the consequence that $u_{N_\qp}(\rho,f)\to0$ for all $\rho$ and $f$.

This normalization problem can be shifted, but not resolved, by allowing instead of $\Co$ a larger algebra $\Calt$ of continuous functions, such as those  going to a constant at infinity (thus $1\in\Calt$), or even all bounded continuous functions. The problem is then that an observable defined in terms of bilinear functionals $u_M$ with $f\in\Calt$ define measures not on the phase space $\Rl^2$, but on a compactification of $\Rl^2$, which can be understood as the set of pure states of $\Calt$. In the examples mentioned these are the one point compactification and the Stone-\v Cech-compactification, respectively. So we have the choice of (a) using $\Co$, for which the set of normalized observables is not compact or (b) using some algebra $\Calt\supset\Co$, for which we may get measures with a positive weight at infinity. The connection of these two points of view is clarified by considering a sequence of observables which converges to zero in the sense of the previous paragraph. The missing normalization of the limit then simply shows up as a positive contribution from the compactification points. Thus we get a normalized observable, but the probability to find a result on ordinary (uncompactified) phase space is zero.  The key point of our proof will thus be to show that this phenomenon cannot happen, provided that both uncertainties are fixed to be finite.

The principles used for calibration and metric uncertainty are rather similar, so we largely treat these cases in parallel. Throughout, we keep the exponents
$1\leq\alpha,\beta<\infty$ fixed. Moreover, we fix some uncertainty levels $\Delta_Q$ and $\Delta_P$, and in the calibration case some parameters $\veps_Q,\veps_P>0$. We then consider the sets $\MM$ and $\MM_c$ of (not necessarily normalized) positive operator valued measures on $\Rl^2$ defined by the membership conditions
\begin{eqnarray}\label{MM}
  M\in\MM&\Leftrightarrow&\quad D_\alpha(M^Q,E^Q)\leq\Delta_Q\quad\mbox{and\ } D_\beta(M^P,E^P)\leq\Delta_P\\
  M\in\MM_c&\Leftrightarrow&\quad \Delta^{\veps_Q}_\alpha(M^Q,E^Q)\leq\Delta_Q\quad\mbox{and\ }
                                   \Delta^{\veps_P}_\beta(M^P,E^P)\leq\Delta_P.  \label{MMc}
\end{eqnarray}
Moreover, we denote by $\NN\subset\MM$ and $\NN_c\subset\MM_c$ the respective subsets of normalized observables. Our aim is to show that these are weakly compact, by first showing that $\MM$ and $\MM_c$ are compact and then that the normalized subsets are closed under weak limits.

\begin{proposition}\label{thm:nonnormalized}
The sets $\MM$ and $\MM_c$ are weakly compact and convex.
\end{proposition}

\begin{proof} The techniques for the two cases are similar, and are based on a description of the respective sets as the sets of (not necessarily normalized) observables satisfying some set of weakly continuous linear constraints derived from \eqref{MM} resp.\ \eqref{MMc}.
We begin with $\MM$, using the Kantorovich dual description of $D_\alpha$ from the equality in \eqref{DKanto}. Including the supremum \eqref{distobs} over states, and using Lemma~\ref{lem:cptsupp} the inequality $D_\alpha(M^Q_\rho,E^Q_\rho)\leq\Delta_Q$ becomes equivalent to the condition that for all $\rho$ and all
$\Psi,\Phi\in {\mathcal C}_0(\Rl)$
satisfying $\Phi(y)-\Psi(x)\leq D(x,y)^\alpha$ we have
\begin{equation}
  \int\Phi(y)\ dM^Q_\rho(y)-\int\Psi(x)\ dE^Q_\rho(x)\leq {\Delta_Q}^\alpha\ .
\end{equation}
Indeed the supremum over $\Phi,\Psi,\rho$ of the left hand side is just $D_\alpha(M^Q,E^Q)^\alpha$. We can further rewrite this as
\begin{equation}
  \int\Phi(y)\chi(p)\ dM_\rho(x,p)\leq {\Delta_Q}^\alpha\ +\int\Psi(x)\ dE^Q_\rho(x),
\end{equation}
where $\chi\in {\mathcal C}_0(\Rl)$ and $0\leq\chi\leq1$. The validity of this inequality for all $\chi,\Phi,\Psi$ with the specified conditions is still equivalent to \eqref{MM}. Moreover, $\Phi(y)\chi(p)\in\Co$, so that left hand side depends on $M$ continuously with respect to the weak topology.
Of course, the momentum part is treated in the same way, together proving compactness of $\MM$. Convexity is obvious, because each of these constraints is linear.

For the calibration case let us write out  the definition of $\MM_c$. The conditions are \begin{eqnarray}\label{cali1}
  \tr\Bigr(\rho (Q-q)^{\alpha}\Bigl)&\leq&\eps_Q^{\alpha}  \quad\Rightarrow
           \int\abs{q'-q}^{\alpha}\,dM_\rho(q',p')\leq\Delta_Q^{\alpha} \\
  \tr\Bigr(\rho (P-p)^{\beta}\Bigl)&\leq&\eps_P^{\beta}  \quad\Rightarrow
            \int \abs{p'-p}^{\beta}\, dM_\rho(q',p')\leq\Delta_P^{\beta} \ .\label{cali2}
\end{eqnarray}
Here $p,q,\rho$ are arbitrary, and the left hand side of these implications (which do not contain $M$) only serve to select a subset of parameters for which the right hand side is to hold. Now the integrals on the right hand side do involve unbounded functions not in $\Co$, so are not directly linear conditions on functionals of the form \eqref{initialtopology}. However, the condition in \eqref{cali1} can equivalently be described as
\begin{equation}\label{continreplace}
  \int f(q',p')\,dM_\rho(q',p')\leq\Delta_Q^{\alpha}\quad\mbox{for all $f\in\Co$ with }
  f(q',p')\leq\abs{q'-q}^{\alpha}.
\end{equation}
With a similar rewriting of the momentum conditions we get inequalities $u_M(\rho,f)\leq\Delta_Q^{\alpha}$ on weakly continuous functionals, so that the set $\MM_c$  is indeed weakly compact.
\end{proof}

We now have to show that the respective normalized sets are closed. The basic idea is to use the fact that there is some unbounded function, which has a uniform finite upper bound on the set of measures under consideration. Therefore, probability cannot ``sneak off to infinity''. This idea (in the case of scalar measures) is made precise in the following Lemma.

\begin{lem}\label{lem:sneak}
Let $(\mu_i)_{i\in I}$ denote a weakly convergent net of probability measures on $\Rl^2$, characterized as positive functionals on $\Co$, which are normalized in the sense that
$\sup_{f\leq1}\int f(x)\,d\mu_i(x)=1$ for all $i$.
Let $h:\Rl^2\to\Rl_+$ be a continuous function diverging at infinity, so that $(1+h)^{-1}\in\Co$, and assume that the expectations of $h$ are uniformly bounded in the precise sense that there is a constant $C$, independent of $i$, such that
$$f\in\Co \ \&\ f\leq h \ \Rightarrow\forall_{i\in I}\  \int f(x)d\mu_i(x) \leq C.$$
Then the weak limit $\mu=\lim_i\mu_i$ is also normalized.
\end{lem}

\begin{proof}
Clearly, the functions $(1+\lambda h)^{-1}$ go to $1$ as $\lambda\to0$. Moreover,
\begin{eqnarray}
  1-\int\frac{d\mu(x)}{1+\lambda h(x)}
      &=&1-\lim_i\int\frac{d\mu_i(x)}{1+\lambda h(x)}
       = \lim_i\sup_{f\leq1}\int\Bigl(f(x)-\frac1{1+\lambda h(x)}\Bigr)d\mu_i(x) \nonumber\\
      &=& \lim_i\sup_{f\leq1}\left\{\lambda \int\frac{f(x)h(x)}{1+\lambda h(x)}\,d\mu_i(x)+ \int\frac{f(x)-1}{1+\lambda h(x)}d\mu_i(x)\right\}\nonumber\\
      &\leq& \lambda C +0
\end{eqnarray}
Here the supremum is over continuous functions $f$ with compact support, so that $fh\in\Co$. Hence, as $\lambda\to0$,  we get $\int(1+\lambda h)^{-1}d\mu\to1$, and $\mu$ is normalized.
\end{proof}

The operator valued version follows in a straightforward way. The normalization operator of a general positive operator valued measure on $\Rl^2$ is, by definition, the operator $M(\Rl^2)$ such that, for all density operators $\rho$
\begin{equation}\label{MX}
  \tr\rho M(\Rl^2)=\sup_{f\leq1}\int f\qp\ dM_\rho(q,p)\ ,
\end{equation}
where the limit is over the increasing net of functions $f\in\Co$ with $f\leq1$. Observables are the normalized operator measures, i.e. those  with $M(\Rl^2)=\idty$. Then we can state the following operator valued version of the previous Lemma:

\begin{corollary}\label{lem:sneakOp}
Let $(M_i)_{i\in I}$ denote a weakly convergent net of observables on $\Rl^2$, and assume that there is a density operator $\rho$ without eigenvalue zero, and a continuous function $h:\Rl^2\to\Rl_+$ diverging at infinity, such that $\int h\qp\ dM_{i,\rho}(q,p)\leq C<\infty$ for a constant independent of $i$. Then the weak limit of the sequence is also normalized.
\end{corollary}

Indeed, we can just apply the previous Lemma to conclude that $\tr\rho(\idty- M(\Rl^2))=0$, which implies $M(\Rl^2)=\idty$ because
$\rho$ has dense range. We note that the same argument holds if the condition is met not for a single $\rho$ but for a family of
states $\rho_k$ with bounds $C_k$ possibly depending on $k$, provided that the union of the ranges of the $\rho_k$ is dense.

\begin{proposition}\label{thm:Ncompact} The sets $\NN$ and $\NN_c$ of observables, defined after Eqs.~\eqref{MM}, \eqref{MMc} are weakly compact.
\end{proposition}

\begin{proof}
In both cases we will apply the Corollary with the same function $h\qp=\abs q^\alpha+\abs p^\beta$.
Now consider $\rho$ to be a Gibbs state of the harmonic oscillator. Its preparation uncertainties $D_\alpha(E^Q_\rho,\delta_0)$ and $D_\beta(E^P_\rho,\delta_0)$ are finite for all $\alpha,\beta<\infty$. Then for any measure in $\NN$ the triangle inequality for the metric implies $D_\alpha(M^Q_\rho,\delta_0)\leq D_\alpha(M^Q_\rho,E^Q_\rho)+D_\alpha(E^Q_\rho,\delta_0)\leq\Delta_Q+D_\alpha(E^Q_\rho,\delta_0)$. Therefore, on $\NN$ we have the uniform bound
\begin{equation}\label{CNN}
  \int h\qp\ dM_\rho(q,p)\leq C=(\Delta_Q+D_\alpha(E^Q_\rho,\delta_0))^\alpha+(\Delta_P+D_\beta(E^P_\rho,\delta_0))^\beta\ ,
\end{equation}
showing that the limit of any weakly convergent sequence from $\NN$ will be normalized according to the corollary. It is also in $\MM$ due to Prop.~\ref{thm:nonnormalized}, hence in $\NN$. It follows that $\NN$ is a closed subset of the compact set $\MM$ and hence compact.

In the calibration case we have to do some additional work, since we have assumptions only about either position calibrating states which are $\veps_Q$-concentrated, or momentum calibrating states which are sharp in momentum. However, from such knowledge we can also infer something about averages of the state over some translations. So let $\rho_Q$ be a Gaussian position calibrating state, say, with $D_\alpha(E_{\rho_Q},\delta_0)\leq \veps_Q$, so that we can conclude $D_\alpha(M^{Q}_{\rho_Q},\delta_0)\leq \Delta_Q$. Consider the phase space translates $\rho_Q\qp$ of these states, which satisfy the calibration condition at the point $q$, and hence
\begin{equation}\label{shiftedcalib}
  D_\alpha(M^Q_{\rho_Q\qp},\delta_q)^\alpha\leq\Delta_Q^\alpha
\end{equation}
Now consider some probability density $f$ on phase space and the state $\rho=f\ast\rho_Q=\int f\qp\rho_Q\qp dqdp$. Then by joint convexity of $D_\alpha$ in its arguments we have $D_\alpha(M^Q_{\rho},f^Q)^\alpha\leq\Delta_Q^\alpha$, where $f^Q$ is the position marginal of $f$. Hence the calibration condition forces
$D_\alpha(M^Q_{\rho},\delta_0)^\alpha\leq[\Delta^Q+D_\alpha(f^Q,\delta_0)]^\alpha$ uniformly with respect to $M\in\NN_c$. A similar relation follows for the averages $g\ast\rho_P$ of momentum calibrating state. What we therefore need to draw the desired conclusion are the following:
A position  calibrating state $\rho_Q$ as described, and a momentum calibrating state $\rho_P$, together with some densities $f,g$ in phase space such that
$\rho=f\ast\rho_Q=g\ast\rho_P$, and this state has no zero eigenvalues. If we take all these objects Gaussian, they are described completely by their covariance matrices and the ``$\ast$'' operation corresponds to addition of covariance matrices. Therefore we can just choose appropriate covariance matrices for $\rho_Q$ and $\rho_P$, and choose $f$ as the Wigner function of $\rho_P$ and $g$ as the Wigner function of $\rho_Q$. The covariance matrix of $\rho$ is then the sum of those for $\rho_Q$ and $\rho_P$, and clearly does not belong to a pure state. Consequently, $\rho$ corresponds to an oscillator state with strictly positive temperature, and hence has no zero eigenvalues. From the estimates given, it is clear that for this $\rho$ a bound of the form \eqref{CNN} holds, so $\NN_c$ is also weakly compact.
\end{proof}

\noindent{\it Summary of proof:\/}\ Applying the Markov-Kakutani fixed point theorem to the transformations $M\mapsto M^{(q,p)}$ acting on the convex compact sets $\NN$ and $\NN_c$, respectively, proves Lemma \ref{crucial-lemma}. Combining this with the results on the covariant case (Prop.~\ref{thm:MURcov}) gives the  Theorem.

\section{Discussion of the constants}\label{sec:constants}
\subsection{Overview}
In this section we give a brief discussion of the constants $c_{\alpha\beta}=c_{\beta\alpha}$ which appear in both the preparation and the measurement uncertainty relations.
Fig.~\ref{fig:ccprime} gives the basic behaviour. The methods for arriving at these plots will be described below.

Since for a probability measure the $\alpha$-norms increase monotonically, we have that $D_\alpha(\mu,\delta_y)$ is increasing in $\alpha$. Hence the constants $c_{\alpha\beta}$ are increasing in $\alpha$ and $\beta$. For every pair of finite values we can use a Gaussian trial state, for which all moments are finite. Therefore, $c_{\alpha\beta}<\infty$. It is interesting to discuss also the limit in which one of the exponents diverges. For $\beta\to\infty$ the $\beta$-norm goes to the $\infty$-norm, i.e., the supremum norm. This is only finite (say $L$) for probability distributions with bounded support, namely the interval $[-L,L]$. The limit $c_{\alpha,\infty}=\lim_{\beta\to\infty}c_{\alpha\beta}$ thus makes a statement how small the $\alpha$-norm of a quantum position distribution can be when the momentum is confined to the interval $[-1,1]$.
As a family of trial states with finite $\alpha$-norm we can take the smooth functions on the interval which vanish with all their derivatives at the boundary. Hence $c_{\alpha\infty}<\infty$ for all $\alpha<\infty$. This case is of interest when particles are prepared by passing through a slit: This will strictly bound the initial position distribution, and hence implies a lower bound on the spread of the momentum distribution, which for free particles is essentially the same as the ballistically scaled position distribution at large times as detected by a far away screen. If the profile of the beam  in the slit is uniform, like a piece of a plane wave, the Fourier transform will be of the form $\sin(x)/x$, for which even the first moment diverges. Hence the uncertainty relations in this case describe how small the lateral divergence of a beam can be made and how to choose the optimal beam profile for that.  The free parameter $\alpha$ allows the optimization to concentrate either on the center or on the tails of the distribution.
Of course, taking both exponents to be infinite is asking too much: There are no Hilbert space vectors which have strictly bounded support in both position and momentum. Hence $\lim_{\alpha\to\infty}c_{\alpha\infty}=\infty$.

\subsection{Hirschman's lower bound}
A good lower bound on $c_{\alpha\beta}$ comes from the work of Hirschman \cite{Hirschman}:  He derived uncertainty relations with general exponents from the entropic uncertainty relations  $H(E^Q_\rho)+H(E^P_\rho)\geq \log(e\pi)$, where $H(\rho)=-\int \rho(x)\log\rho(x)\, dx$ denotes the Shannon entropy of a probability distribution $\rho$ with respect to Lebesgue measure. At the time of Hirschman's work, the best constant in this in equality was only conjectured, and proven only later in \cite{Beckner}. We recapitulate Hirschman's nice argument, if only to free it of the unnecessary restriction $\alpha=\beta$ which he made. This was lifted in \cite{Cowling}, although these authors were more interested in further generalizations than in finding tight constants. The basic idea of Hirschman is to use variational calculus to maximize the entropy among all probability distributions with given $\alpha^{\rm th}$ moment $M_\alpha=\int \abs x^\alpha\rho(x)\,dx$. This gives a probability density proportional to $\exp(-\zeta\abs x^\alpha)$, where the Lagrange multiplier $\zeta$ can be determined explicitly from $M_\alpha$. It turns out that for the maximizing state the entropy is of the form $(1/\alpha)\log M_\alpha=\log D_\alpha(\rho,\delta_0)$ plus an additive constant $A(\alpha)$ independent of $M_\alpha$. Since for an arbitrary distribution with this moment the entropy must be lower, we get the inequality
\begin{eqnarray}\label{Shannon}
  H(\rho)&\leq& \log D_\alpha(\rho,\delta_0)-A(\alpha) \\
  A(\alpha)&=& \frac{\log(e\alpha)}\alpha+\log\bigl(2\Gamma(1-\frac1\alpha)\bigr)\ .
\end{eqnarray}
Since the entropy does not change under translation, we may also replace the point zero in \eqref{Shannon} by any other one, like the one minimizing \eqref{Delalpha}. That is, we may replace $D_\alpha$ by the $\alpha$-spread. Combining the entropic uncertainty relation with \eqref{Shannon} then gives
\begin{eqnarray}\label{Hirsch1}
  \log(\Delta_\alpha(E^Q_\rho) \Delta_\beta(E^P_\rho))&\geq& H(E^Q_\rho)+H(E^P_\rho)+A(\alpha)+A(\beta)\nonumber\\
         &\geq&\log(e\pi)+A(\alpha)+A(\beta)=\log c^H\sab
\end{eqnarray}
with the constants
\begin{equation}\label{Hirsch2}
  c^H\sab=\frac{\pi\exp(1-\frac1\alpha-\frac1\beta)\,\alpha^{1/\alpha}\beta^{1/\beta}}
                     {4\Gamma(1+\frac1\alpha)\Gamma(1+\frac1\beta)}\ .
\end{equation}
As Fig.~\ref{fig:Hirsch} shows, the Hirschman bound is quite good and exact only at $\alpha=\beta=2$. Indeed, for the bound to be tight the probability densities would have to be $\rho_Q(x)\propto\exp(-\zeta\abs x^\alpha)$ and $\rho_P(p)\propto\exp(-\zeta\abs p^\beta)$. But this is not compatible with a Fourier pair, except when both exponents are 2. Another feature which is missed by Hirschman's bound is the divergence when both exponents become infinite. Indeed, from the point of view of entropic uncertainty, there is no obstruction against both momentum and position being compactly supported.

\begin{figure}[ht]
\centering
  \includegraphics[width=6cm]{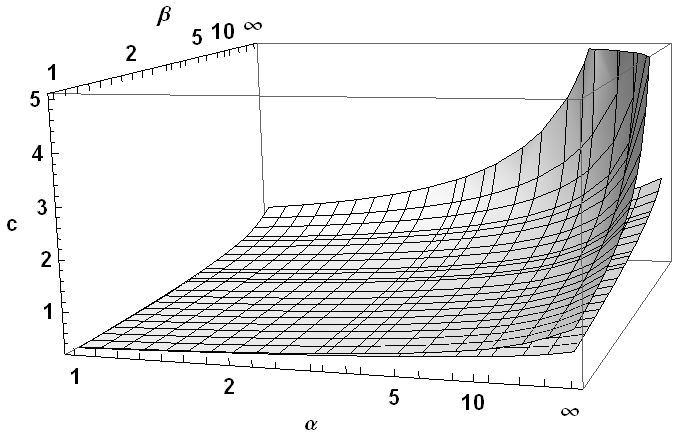}
  \includegraphics[width=6cm]{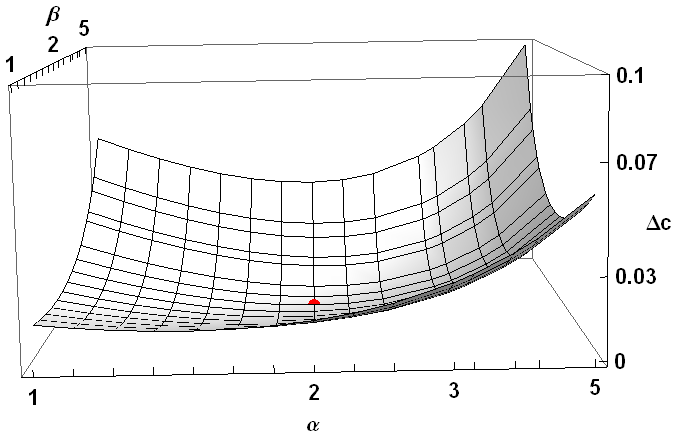}
	\caption{Left panel: Hirschman's lower bound in comparison with the exact bound. Scaling of axes as in Fig.~\ref{fig:ccprime}.
             Right panel: difference of the two functions around $\alpha=\beta=2$ (marked), where the bound is exact.  }
	\label{fig:Hirsch}
\end{figure}

The Hirschman techniques gives us no handle on estimating the first excited state $g'\sab$.

\subsection{Exactly solvable cases}
Let us now turn again to the optimal bounds. The ground state problem for $H(\alpha,\beta)$ can be solved in closed form (up to the solution of an explicitly given transcendental equation) only for a few special values. When $\beta=2$ it is a standard Schr\"odinger operator for some anharmonic oscillator, which is harmonic for $\alpha=2$, leading to the well-known value $c_{22}=\frac12$. For $\alpha=\infty$ we get a particle in the box $[-1,1]$, for which the ground state is $\psi(x)=\cos(\pi x/2)$. This leads to $c_{2,\infty}=\pi/2$. At the other end we have the potential $\abs Q$, which we can consider as a linear potential on the half line with Neumann boundary conditions ($\psi'(0)=0$). This is solved by $\psi(x)={\Airy}(x-\lambda)$, where $\Airy$ is the Airy function, and $-\lambda\approx1.0188$ is the first zero of the derivative $\Airy'$, which is also the eigenvalue. Hence $c_{12}\approx0.3958$. Finally, in the case $\alpha=2m$, $m\in\Nl$, $\beta=\infty$, $H$ is a differential operator on the interval $[-1,1]$ with constant coefficients. Then $\braket\psi{H\psi}$ is finite if $\psi$ is in the domain of $P^m$ considered as an operator on the whole line. This entails that the derivatives up to order $m-1$ are continuous, and hence vanish at the boundary. Unfortunately, the equations characterizing the linear combinations of exponential functions satisfying the boundary conditions are transcendental with complexity increasing rapidly with $m$. For $\alpha=4$ we have to solve the equation $\tan\gamma=-\tanh\gamma$ for the eigenvalue $\gamma^4$, and $c_{4\infty}=\gamma\approx2.365$. For $\alpha=6$ the relevant solution of the fairly complicated equation is exactly $\pi$, and $c_{6\infty}=\pi$. Similarly, and with reasonable effort, one can get $c_{8\infty}$ and $c_{10,\infty}$.

In all cases described here, the excited states, and in particular the first, $g'\sab$ can be obtained in the same way.
Table~\ref{table:exact} summarizes the results.

\begin{table}
\let\flush\hfill
\begin{tabular}{|c|c||c|c|}\hline
$\alpha$&$\beta$&$c\sab$&$c'\sab$\\ \hline
$1$&$2$&\flush$            0.396$&\flush$1.376$\\
$2$&$2$&\flush$        1/2=0.500$&\flush$3/2=1.500$\\
$2$&$\infty$&\flush$ \pi/2=1.571$&\flush$\pi=3.142$\\
$4$&$\infty$&\flush$       2.365$&\flush$3.927$\\
$6$&$\infty$&\flush$   \pi=3.142$&\flush$4.714$\\
$8$&$\infty$&\flush$       3.909$&\flush$5.498$\\
$10$&$\infty$&\flush$      4.672$&\flush$6.279$\\
$\infty$&$\infty$&$ \infty$&$\infty$\\\hline
\end{tabular}
\caption{Exactly solvable values of the optimal lower bound.}\label{table:exact}
\end{table}

\subsection{Expansion in oscillator basis}
For general exponents a numerical approach which works well for small $\alpha,\beta$ and for $\alpha\approx\beta$ is to compute the matrix elements of $H(\alpha,\beta)$ in the harmonic oscillator basis, truncated at some level $n$, and to compute the ground state of the resulting matrix. Since already the coherent bound (corresponding to $n=0$) is fairly good, even small $n$ gives a fairly good approximation. The computation of the matrix elements can be done exactly (in terms of $\Gamma$-functions), so the numerical error is practically only on the truncation. The results are shown in Fig.~\ref{fig:coherent}.
\begin{figure}[ht]
\centering
  \includegraphics[width=12cm]{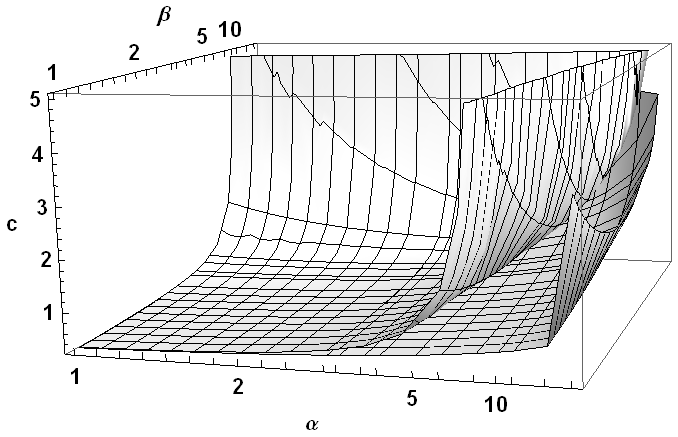}
	\caption{Comparison of the coherent state uncertainty product (top) and the minimized product using oscillator eigenfunctions up to $n=200$ (bottom).
     Axes scaling as in Fig.~\ref{fig:ccprime}.}
     \label{fig:coherent}
\end{figure}

\subsection{One index infinite}
It is apparent from Fig.~\ref{fig:coherent} that for high exponents the approximation in terms of oscillator eigenfunctions becomes unreliable. The case of one infinite exponent is again easier to handle, because instead of a high exponent one just has to implement a support condition. It turns out that for $\beta=\infty$ and all $\alpha$ a good first approximation is the wave function
\begin{equation}\label{PsiInfTest}
  \Hat\Psi(p)\propto(1-p^2)^\alpha_+
\end{equation}
where $x_+$ denotes the positive part of $x\in\Rl$ (i.e., $x_+=x$ for $x\geq0$ and $x_+=0$ for $x\leq0$). The $\alpha^{\rm th}$ moment of the position distribution can be evaluated explicitly giving the bound
\begin{equation}\label{cInfTest}
  c_{\alpha\infty}\leq \left(\frac{\Gamma \left(\frac{\alpha }{2}+1\right) \Gamma \left(\alpha
   +\frac{3}{2}\right)}{\Gamma \left(\frac{\alpha
   +3}{2}\right)}\right)^{\frac{1}{\alpha }}
   =\frac\alpha e+ \frac{\ln(4\pi\alpha)}{2e}+ {\mathbf o}(1), \quad \mbox{as }\ \alpha\to\infty,
\end{equation}
where the second expression is the Stirling approximation.
\begin{figure}[ht]
\centering
  \includegraphics[width=7cm]{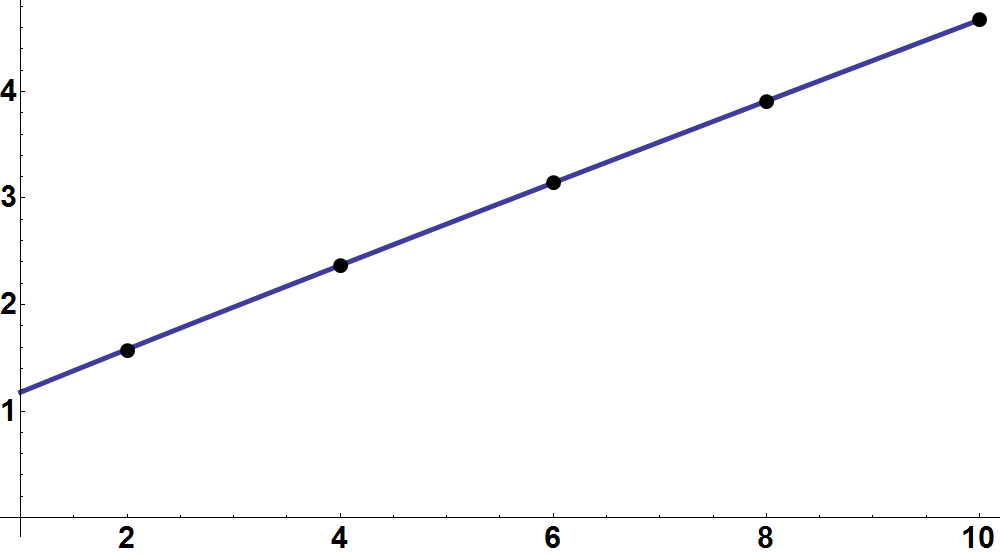}
	\caption{The upper bound \eqref{cInfTest}. The dots represent the known exact values of $c_{2n,\infty}$ cited above.  }
	\label{fig:inftrial}
\end{figure}

One can improve this in a similar way as for the Gaussian trial function, by multiplying \eqref{PsiInfTest} with polynomials in $p$, i.e., by expressing the Hamiltonian in terms of associated Legendre functions. This confirms the close approximation shown in Fig.~\ref{fig:inftrial} also for the non-integer values, with an error decreasing exponentially.

\section{Extensions and generalizations}\label{sec:extensions}
\subsection{Other Observables}
It is clear that the basic definitions of errors can be applied to arbitrary observables. An appropriate metric has to be chosen on the outcome space of each observable. Then, whenever two observables $A,B$ are not jointly measurable, there will be a measurement uncertainty relation, which expresses quantitatively that $D_\alpha(M^A,E^A)$ and $D_\beta(M^B,E^B)$ cannot both be small. For the analogous statement of calibration errors it is needed that $A$ and $B$ are projection valued. As in the case of preparation uncertainty relations there may be many ways of expressing mathematically that $\Delta A$ and $\Delta B$ ``cannot both be small''.

The product form $\Delta A\ \Delta B\geq c$ is a rather untypical expression of this sort, which is specific to canonical pairs $(Q,P)$ and their dilation symmetry
$Q\mapsto\lambda Q$; $P\mapsto P/\lambda$. One could also say: The product form is fixed by dimensional analysis. For general $A$,$B$ one should think of uncertainty trade-offs in terms of an  ``uncertainty diagram'' describing the set of pairs $(\Delta A,\Delta B)$ realizable by appropriate choice of preparation or approximate joint measurement (see Fig.~\ref{fig:finUR} 
for an illustrative example). An ``uncertainty relation'' would be any inequality that excludes the origin $\Delta A=\Delta B=0$ and some region around it. Of course, contrary to the entire textbook literature the Robertson form $\Delta A \Delta B\geq\frac 12\abs{\langle[A,B]\rangle}$ (like Schr\"odinger's improvement \cite{Schro1930}) is {\it not} an uncertainty relation in this sense, due to the state dependence of the right hand side. In fact the {\it only} cases in which the best constant in an ``uncertainty relation'' of product form is positive are canonical pairs. In contrast, a relation of the form   $(\Delta A)^2+(\Delta B)^2\geq c^2$ can always be used to make a non-trivial statement, even if this does not capture the full story contained in the uncertainty diagram.

For any pair of observables $A,B$, and choices of metrics and exponents, we now have two uncertainty diagrams: one for preparation uncertainty and one for measurement uncertainty. It is a very special feature of the canonical case studied in this paper, perhaps due to the very high symmetry, that the two diagrams coincide. In general they will be different. Indeed (for sharp observables) the origin is included in the preparation uncertainty diagram if and only if $A$ and $B$ have at least {\it one} common eigenvector, whereas the measurement uncertainty diagram contains the origin if and only if the observables commute, and hence have {\it a basis} of common eigenvectors. An example in the opposite direction is given by a pair of jointly measurable unsharp observables for which no output distribution is ever concentrated on a point.
 This would leave the logical possibility that for sharp observables the preparation uncertainty diagram is always included in the measurement diagram, but much too little is known about either kind of uncertainty to even offer this as a conjecture.

\subsection{General phase spaces, including finite ones }
The methods in this paper do extend to discrete canonical pairs, i.e., to pairs of unitaries $U,V$ which commute up to a root of unity ($UV=\exp(2\pi i k/d)VU$). The observables in question are then the spectral measures of $U$ and $V$. The irreducible representations of this relation (the analogue of the Schr\"odinger representation studied in this paper) are $d$ dimensional, with ``position'' $U$ represented as a multiplication operator and ``momentum'' $V$ the cyclic shift by $k$ steps. Further generalizations allow any locally compact abelian group $X$ to replace the cyclic group $X$ in this example, with the position observable on $\LL^2(X)$ and momentum generated by the shifts, corresponding by a Fourier transform to an observable on the dual group $\widehat X$. A joint measurement of these thus has the outcome space $X\times\widehat X$, also called phase space. The case $X\cong\widehat X\cong\Rl$ leads back to standard phase space, $X\cong\widehat X\cong\Ir_d$ is the cyclic case. However, this class also contains the ``Fourier series'' case $X=\{e^{it}|t\in(-\pi,\pi]\}$, $\widehat X=\Ir$, and arbitrary products of all these examples, like the phase spaces for quantum systems with many canonical degrees of freedom.

Then the methods of this paper apply, with the following modifications:
\begin{itemize}
\item One has to choose a translation invariant metric on each of the spaces $X$ and $\widehat X$. On non-compact groups $X,\widehat X$ it should have compact level sets, and hence diverge at infinity.
\item The harmonic analysis \cite{QHA} sketched in Remark 2 carries over \cite{Schultz2013}. In particular, all covariant phase space measurements are parameterized by density operators, and their marginals are formed by convolution as in Eq.~\eqref{marginalconcolve}.
\item The properties of Wasserstein metrics under convolution were already considered at the required level of generality in Sect.~\ref{sec:errors}.
\item The averaging argument carries over, with the only modification that the compactness discussion becomes superfluous in the finite case.
\item Hence our main result holds in the form that for any $X$, any choice of metrics, and any $\alpha,\beta$ the uncertainty diagrams
   for (a) preparation uncertainty, (b) measurement uncertainty according to calibration criteria and (c) measurement uncertainty according to Wasserstein metrics coincide.
\end{itemize}
To be precise we have shown only that the monotone closures of these diagrams coincide, i.e., the diagrams in which we only care how small uncertainties can be, so that with every point also the positive quadrant above it is included, and white spaces as those near the axes in Fig.\ref{fig:finUR} are filled in. If we restrict to covariant measurements the diagrams would coincide even without the monotone closure, since the measurement uncertainties are just equal to suitable preparation uncertainties. However, the averaging argument gives only that for every pair of measurement uncertainties there is a pair of, in general, smaller preparation uncertainties, so the monotone closure is needed.

\begin{figure}[ht]
\centering
  \includegraphics[width=6cm]{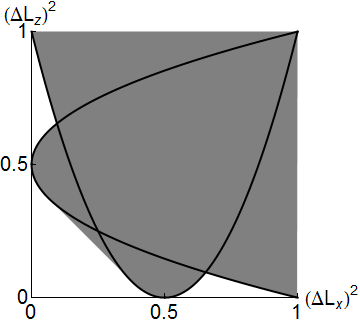}\hskip 1 cm
  \includegraphics[width=5.8cm]{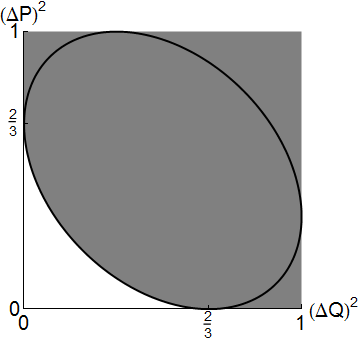}
	\caption{Left panel: Preparation uncertainty diagram for two angular momentum components of a Spin-1 system. Boundary lines are parabolas indicated and, partly, their convex hull. Results from \cite{Sachse}.\\
      Right panel: Uncertainty diagram for discrete canonical variables and discrete metric. It simultaneously represents the uncertainties for preparation, as well as measurement using either the metric criterion or the calibration criterion. The boundary line is part of the ellipse indicated. The diagram is drawn for dimension $d=3$. }
	\label{fig:finUR}
\end{figure}

Rather than displaying a zoo of uncertainty diagrams we consider here just the case of a finite cyclic group $X=\Ir_d$. For finite outcome sets it is often natural to choose a metric which makes the uncertainty criteria independent of a relabelling of the outcomes, like entropic uncertainty relations. This forces the discrete metric $D(x,y)=(1-\delta_{xy})$. Since then $D(x,y)^\alpha$ is independent of $\alpha$ the distance functions $D_\alpha(\mu,\nu)$ all express the same quantity, which turns out to be essentially the variation norm:
\begin{eqnarray}\label{discretenorm}
  D_\alpha(\mu,\nu)&=&\left(\frac12\norm{\mu-\nu}_1\right)^{1/\alpha}  \\
  \Delta_\alpha(\mu)&=&1-\max_x \mu(\{x\})  \label{discreteVar}
\end{eqnarray}
The proof of \eqref{discretenorm} is easiest by using the dual characterization of $D_1$ as a supremum over functions with Lipshitz constant $1$, and noting that this is equal to the supremum over all functions $f:X\to[0,1]$. Consider now a density operator $\rho$ on $\ell^2(X)$. Its position distribution is given by the expectations of the projections $\kettbra x$, and by translation invariance it suffices to consider one of them, say $\psi^Q=\ket0$. Similarly, the momentum probabilities are given by the expectation of $\kettbra{\psi^P}$ with $\brAket{\psi^P}x=1/\sqrt d$  and its momentum translates. Therefore the uncertainty diagram is
\begin{equation}\label{discretedia}
  \bigl\{(\Delta_1(E^Q_\rho),\Delta_1(E^P_\rho)\bigr\}
    =\bigl\{(1-\brAAket{\psi^Q}\rho{\psi^Q},\  1-\brAAket{\psi^P}\rho{\psi^P})\bigr\}
\end{equation}
where $\rho$ runs over all density operators. Clearly this depends only on the restriction of $\rho$ to the two dimensional subspace generated by $\psi^Q$ and $\psi^P$. When $d=2$ this is the whole space and the diagram is a section of the Bloch sphere in some slanted coordinates, i.e., an ellipse. For higher dimensions we get, in addition the point $(1,1)$ and all segments connecting the ellipse to this point. The monotone closure is always the same. It is easy to see that the ellipse is centered at the point $(\frac12,\frac12)$, and touches the axes at $1-1/d$. This completely fixes the diagram (right panel in Fig.~\ref{fig:finUR}).

We have covered here only the uncertainty relations which come out of our analysis practically without additional work. Of course, there are many other pairs of observables one would be interested in. For some studies in this direction we recommend \cite{Carmeli_etal2011,Buscemi,Coles2013}.

Finally, we note that the special case of the qubit observables has extensively been studied in a separate paper \cite{BLW2013qubit}, where additive error trade-off relations are proven that can be tested by the same experiments performed to test an inequality due to Ozawa.

\subsection{More state dependence?}
It is clear that the inequality $D_\alpha(M^Q_\rho,E^Q_\rho)D_\beta(M^P_\rho,E^P_\rho)\geq c$, if it were true for all $\rho$, would be a much stronger Theorem than ours, which claims only a relation for the suprema of each of the factors. However, such a relation trivially fails, for example, by choosing for $M$ an ideal position measurement plus the random generation of a momentum output drawn according to $E^P_\rho$. Rather than touting this as a refutation of Heisenberg's paper, one can look for true relations which are intermediate between the state dependent one and the double supremum considered in this paper. A natural candidate is
\begin{equation}\label{jointsup1}
  \inf_M\sup_\rho D_\alpha(M^Q_\rho,E^Q_\rho)D_\beta(M^P_\rho,E^P_\rho)
    \leq \inf_M\sup_{\rho,\sigma} D_\alpha(M^Q_\rho,E^Q_\rho)D_\beta(M^P_\sigma,E^P_\sigma)=c\sab.
\end{equation}
Note that by the argument given above, switching $\inf$ and $\sup$ on the left hand side would again trivially produce zero.

The coupled supremum is difficult to compute. David Reeb \cite{Reeb} has evaluated at least a restricted version of it, namely for $\alpha=\beta=2$ and both $M$ and $\rho$ Gaussian. Indeed this can be done in a straightforward way using Example~\ref{ex:Wass2}.
Let us take $\rho$ as centered and with spreads $r_Q$ and $r_P$. Take $s_Q,s_P$ be the corresponding ones for the likewise Gaussian state $\sigma$ defining $M$. Then we can use the explicit form of the Wasserstein-2 metric to get
$$D_2(M^Q_\rho,E^Q_\rho)=\sqrt{r_Q^2+s_Q^2}-r_Q=r_Q\Bigl(\sqrt{1+(s_Q/r_Q)^2}-1\Bigr)$$
By concavity of the square root, this is bounded above by $s_Q^2/(2r_Q)$, and because $r_Qr_P\geq(1/2)\hbar=s_Qs_P$, the uncertainty product is bounded by $(1/4)(\hbar/2)$. This bound is not tight. Since $f(x)=\sqrt{1+x^2}-x$ is decreasing, the maximum is taken on the minimal uncertainty $r_Q,r_P$, i.e., the maximum over all Gaussian inputs is the maximum of $f(x)f(1/x)$, which is attained at $x=1$. This translates into $r_Q=s_Q$ and $r_P=s_P$, so the maximum over all Gaussian inputs is
\begin{equation}\label{jointsup2}
  \sup_{\rho\ \mbox{\small Gaussian}}\ D_2(M^Q_\rho,E^Q_\rho)D_2(M^P_\rho,E^P_\rho)=(\sqrt2-1)^2\frac\hbar2\approx .17 \frac\hbar2.
\end{equation}
Thus it appears that there might be a proper gap in \eqref{jointsup1}, but the evidence is rather indirect. One should also point out that while the double sup version has a straightforward interpretation coupling two figures of merit, it is not so clear what the coupled sup would be telling us.

\subsection{Finite operating ranges}
The figure of merit obtained by taking the worst case over {\it all} input states is very demanding indeed. In practice, for assessing the performance of a microscope we would not worry about the resolution on objects light years away. Therefore it is reasonable to restrict the supremum to states localized in some finite {\it operating range}. Shrinking this range to zero would bring us essentially back to the state dependent approach. For a good microscope the operating range should at least be large compared to the resolution. What we considered in this paper is the idealization in which this ratio goes to infinity.

We do plan to make this explicit, and set up uncertainty relations also with finite operating ranges, which in the limit converge to the ones given in this paper. In fact, for the squared noise operator approach this has already been considered by Appleby \cite{Appleby1998b}.

\subsection{Entropic versions}
Shortly before this paper was completed, a related paper \cite{Buscemi} on entropic state independent noise-disturbance uncertainty relation appeared, providing a kind of entropic version of the idea of calibration.
 In this paper the noise  $N(\h M,A)$
in an approximate measurement $\h M$ of a discrete sharp nondegenerate finite-level observable $A$
is  quantified (in the form of entropy) by how well it is possible to guess from the measurement outcome distributions the input eigenstate $\rho_k$ from a uniform distributions of such inputs.
Similarly, the unavoidable disturbance $D(\h M,B)$ quantifies (in the form of entropy) the extent to which the action of $\h M$ necessarily reduces the information about which eigenstate $\sigma_l$ of  the observable $B$, another discrete sharp nondegenerate observable,  was initially chosen among a uniform distribution of them. Using the Maassen-Uffink entropic uncertainty relation for preparations \cite{Maassen_1988}, the entropic noise-disturbance trade-off relation then takes the additive form
$$
N(\h M,A)+D(\h M,B)\geq - \log c,
$$
where $c=\max_{k,l}\tr{\rho_k\sigma_l}$ is the same constant as in the preparation relation. It seems to be an open question if there is a measurement which saturates the inequality.  It remains to be seen how this approach extends beyond the finite-dimensional case.

\section{Conclusion and Outlook}\label{sec:conclusion}

We have formulated and proved a family of measurement uncertainty relations for canonical pairs of observables. This gives one possible rigorous interpretation of Heisenberg's 1927 statements.

The particular case of canonical variables is special, due to the phase space symmetry of the problem. This leads to the complete equivalence of the possible values of $(\Delta Q,\Delta P)$ between preparation and measurement uncertainty, even when the exponents $\alpha,\beta$ are varied. In order to establish this we had to generalize standard preparation uncertainty relation to general power means as well, and gave a characterization of the optimal constants in terms of a ground state problem to be solved numerically.

\section{Acknowledgements}
This work is partly supported (P.B., P.L.) by the Academy of Finland, project no 138135, and EU COST Action MP1006.
R.F.W.\ acknowledges support from the European network SIQS. This work has benefited from discussions at a conference organised as part of COST Action MP1006, {\em Fundamental Problems in Quantum Physics}. RFW acknowledges the hospitality of the program ``Mathematical Horizons in Quantum Physics 2013'' at the IMS Singapore, and in particular discussions with David Reeb, who pointed out Ref.~28, Example~\ref{ex:Wass2} and its use to compute the combined supremum.

\bibliographystyle{unsrtnat}


%

\end{document}